\documentclass[lettersize,journal]{IEEEtran}

\renewcommand{\figurename}{Fig.} 

\usepackage{caption}
\usepackage{graphicx}
\usepackage{float} 
\usepackage{subfigure}
\usepackage{subcaption}
\usepackage{color}

\usepackage{amssymb}
\usepackage{graphicx}
\usepackage[utf8]{inputenc} 
\usepackage[T1]{fontenc}
\usepackage{url}
\usepackage{ifthen}
\usepackage{cite}
\usepackage[cmex10]{amsmath} 
\usepackage{amsmath,amsthm}

\usepackage{algorithm}
\usepackage{algpseudocode}

\newtheorem{thm}{Theorem}
\newtheorem{defn}{Definition}

\newtheorem{corollary}{Corollary}
\newtheorem{example}{Example}
\newtheorem{conj}{Conjecture}

\newtheorem{proposition}{Proposition}

\newtheorem*{remark*}{Remark}

\begin{document}

\title{Stopping Set Analysis for Concatenated \\ Polar Code Architectures}

\author{
Ziyuan~Zhu,~\IEEEmembership{Student Member,~IEEE,}
and~Paul~H.~Siegel,~\IEEEmembership{Life~Fellow,~IEEE}

\thanks{Manucript received MM DD, YYYY; revised MM DD, YYYY. The associate editor coordinatitng the review of this paper and appriving it for publication was NAME.}
\thanks{The authors are with the Electrical and Computer Engineering Department, University of California San Diego, La Jolla, CA  92093, USA (e-mail: \{ziz050, psiegel\}@ucsd.edu).}
\thanks{The material in this paper was presented in part at the IEEE International Symposium on Information Theory (ISIT), 2024.}
\thanks{This work was supported by the National Science Foundation under grants CCF-1764104 and CCF-2212437.}
\thanks{Digital Object Identifier xx.xxxx/xxxxxxx.xxxx.xxxx}

}


\maketitle

\begin{abstract}
This paper investigates properties of concatenated polar codes and their potential applications. We start with reviewing previous work on stopping set analysis for conventional polar codes, which we extend in this paper to concatenated architectures. Specifically, we present a stopping set analysis for the factor graph of concatenated polar codes, deriving an upper bound on the size of the minimum stopping set. To achieve this bound, we propose new bounds on the size of the minimum stopping set for conventional polar code factor graphs. The tightness of these proposed bounds is investigated empirically and analytically. We show that, in some special cases, the exact value of the minimum stopping set can be determined with a time complexity of $O(N)$, where $N$ is the codeword length. The stopping set analysis motivates a novel construction method for concatenated polar codes.  
This method is used to design outer polar codes for two previously proposed concatenated polar code architectures: augmented polar codes and local-global polar codes. Simulation results demonstrate the advantage of the proposed codes over previously proposed constructions based on density evolution (DE). 
\end{abstract}


\IEEEpeerreviewmaketitle

\section{Introduction}

\IEEEPARstart{P}{olar} codes, introduced by E. Ar\i kan~\cite{Arikan2009}, occupy a unique place in the history of error correction codes as the first family of codes to achieve the Shannon capacity of arbitrary binary symmetric memoryless channels (BSMs). The code construction starts from a \textit{channel transformation}, where $N$ synthesized bit-channels $W_N^{(i)}$, $i{=}0,1,...,N{-}1$ are obtained by applying a linear transformation to $N$ independent copies of a BSM channel $W$. As the block length $N$ goes to infinity, the synthesized bit-channels become either noiseless or completely noisy. A polar code carries information on the least noisy bit-channel positions and freezes the remaining ones to a predetermined value, usually chosen to be zero. Ar\i kan~\cite{Arikan2011} also introduced the concept of systematic polar encoding, achieved through the solution of linear encoding equations that ensure the codewords contain the information bits at designated positions. 

Concatenated polar codes have been proposed that leverage the error floor performance of polar codes in conjunction with other powerful codes such as Low-Density Parity-Check codes~\cite{Eslami2013} and Reed-Solomon codes~\cite{Bakshi2010}. Expanding upon the enhanced belief propagation construction of Guo et al.~\cite{Guo2014}, Elkelesh et al.~\cite{Elk2017}  introduced an augmented polar code architecture that concatenates two polar codes, using an outer auxiliary polar code to further protect the semipolarized bit-channels within the inner polar code. In the same work, they also suggested connecting several inner polar codes through a single auxiliary polar code, offering the flexibility of codeword lengths other than a power of two. Motivated by practical applications in data storage and low-latency communication systems, Zhu et al.~\cite{Zhu2022}  proposed an architecture for polar codes offering local-global decoding. In this scheme, a codeword comprising several inner polar codes is concatenated with a systematic outer polar code, thus enabling both local decoding of  the inner codes and global decoding of the codeword.

The belief propagation (BP) decoder for polar codes was introduced to increase throughput through a pipelined decoding process~\cite{ArikanBP}. While the BP decoder surpasses the error rate performance of the original successive-cancellation (SC) decoder, it falls short of the SC-list (SCL) decoder~\cite{SCL}. The BP-list (BPL) decoder~\cite{BPL}, which incorporates different permutation patterns of BP decoding units, significantly enhances error rate performance, bridging the performance gap between BP-based and SC-based decoders. 

Polar codes and Reed-Muller (RM) codes share the same basic encoding matrix before selecting the information set: RM codes select rows according to their Hamming weights, while polar codes select rows by comparing their associated Bhattacharyya parameters~\cite{Arikan2009}. Another frozen set selection method, introduced by Mori et al.~\cite{DE}, uses density evolution (DE) to analyze  BP results for each decoding tree corresponding to the SC decoding process. The high computational complexity of DE motivated the Gaussian approximation (GA) algorithm~\cite{GA},  which assumes that the log-likelihood ratio (LLR) distribution corresponding to each variable node is a Gaussian with mean $m$ and variance $\sigma^2=2m$, thus reducing the convolution of densities to a one-dimensional computation of mean values. In~\cite{GA Polar}, Dai et al. proposed a modification to GA to address the performance loss incurred when applying GA to long polar codes.

An important characteristic of polar codes is that the bit-channel orderings are channel-dependent. Although no general rule is known for completely ordering the bit-channels of a general BSM channel, some partial orders (POs) that are independent of the underlying channel have been found for selected bit-channels~\cite{DE},~\cite{PO2016},~\cite{Bardet2016}. In~\cite{DE}, an ordering applicable to bit-channels with different Hamming weights was presented. The Hamming weight of $W_N^{(i)}$ is defined as the number of ones in the binary expansion of $i$. The ordering  states that a bit-channel $W_N^{(j)}$ is stochastically degraded with respect to $W_N^{(i)}$ if the positions of 1 in the binary expansion of $j$ are a subset of the positions of 1 in the binary expansion of $i$. The ordering in~\cite{PO2016} and~\cite{Bardet2016} compared bit-channels with the same Hamming weight. It was based on the observation that a bit-channel $W_N^{(j)}$ is stochastically degraded with respect to $W_N^{(i)}$ if $j$ is obtained by swapping a more significant 1 with a less significant 0 in the binary expansion of $i$. Both of these orderings are partial, in the sense that not all bit-channel pairs $(W_N^{(i)},W_N^{(j)})$ are comparable. A more general investigation of POs for polar codes can be found in~\cite{Wei2019}.

While design methods based on the Bhattacharyya parameters, DE, and GA were originally used in the context of SC decoding, they have also been applied to code design for BP decoding. Eslami et al.~\cite{Eslami2013} introduced a construction method based on stopping sets in the sparse polar code factor graph, aimed at increasing the stopping distance of the polar code. They provided empirical evidence showing improved performance under BP decoding, compared with the conventional code design.

In this paper, we study the stopping sets within concatenated polar code architectures. The analysis of stopping sets in the concatenated factor graph
suggests a novel code construction method that identifies  promising  information sets for the outer code. 
Error rate simulations demonstrate that the proposed method can improve the performance of augmented and local-global polar codes.
Portions of this paper were presented in~\cite{Zhu2024}.

The paper is organized as follows. Section II briefly reviews background results and notations used in the rest of the paper. In Section III, we provide the stopping set analysis for concatenated polar codes, and we emphasize the importance of finding the minimum stopping set within the conventional polar code factor graph that includes certain information nodes. Section IV presents lower and upper bounds on the size of these minimum stopping sets, while Section V provides exact calculations for specific information node choices on the leftmost stage of the polar code factor graph. In Section VI, we propose an outer code design method based on stopping set analysis for concatenated polar code architectures. Finally, Section VII concludes the paper.

\section{Preliminaries}

\subsection{Polar Codes and Systematic Polar Codes}

In conventional polar code design, $N$ independent copies of a channel $W$ are combined in a recursive manner into a vector channel $W_N$, which is then split into $N$ channels ${W_N^{(i)}, \;0\leq{i}\leq{N-1}}$, referred to as bit-channels.  The Bhattacharyya parameter $Z(W_N^{(i)})$ is used to identify the quality of bit-channel $i$. A polar code of rate $R{=}\frac{K}{N}$ selects the $K$ most reliable bit-channels (with the smallest $Z(W_N^{(i)})$) to input information bits, and the remaining bit-channel inputs are frozen to zero. We use ${\mathcal A}$ to denote the set of information indices, and ${\mathcal F{=}\mathcal A^c}$ to denote the frozen indices. Let $G{=}F^{\bigotimes n}$ be the $N{\times} N$ matrix that is the $n$-th Kronecker power of $F{=}\left[\begin{array}{cc}1 & 0 \\1 & 1 \\ \end{array} \right]$, where $n{=}\log_2 N$.  The polar encoding process is specified by  $x{=}uG$, where $ x,u\in \mathbb{F}^N, G\in \mathbb{F}^{N\times N} $.  

Ar\i kan~\cite{Arikan2011} showed that a systematic encoder can be realized that maps information bits to positions in the set ${\mathcal B} {=} {\mathcal A}$ in the codeword $x$. To be specific, $u_{\mathcal A^c}$ is set to 0,  $x_\mathcal B$ is set to the information vector, and  $u_\mathcal A$ and $x_{\mathcal B^c}$ are found by solving a system of equations.

\subsection{Concatenated Polar Codes}

Our focus in this paper is on concatenated code architectures in which all component codes are polar codes. 
The augmented and flexible length architectures were introduced in~\cite{Elk2017}. In an augmented polar code, a short, rate $R_0=\frac{K_0}{N_0}$  auxiliary outer polar code $G_0$ is connected to an inner polar code $G_1$ of length $N_1$. The $N_0$ bits of the outer codeword are assigned to the semipolarized bit-channels within the inner code (through an interleaver). An additional $K_1$ information bits are assigned to the good bit-channels within the inner code. The total code rate for the augmented structure is $R_{aug} = \frac{K_0+K_1}{N_1}$.

In the flexible length architecture, two inner codes $G_1, G_2$ of length $N_1, N_2$ are coupled through a rate $R_0=\frac{K_0}{N_0}$ auxiliary outer code $G_0$. Information words of length $K_1, K_2$ are assigned to the good bit-channels of the two inner codes, respectively.
The outer codeword  is divided into two parts which are assigned to the semipolarized bit-channels of the inner codes. The total encoding rate for the flexible length structure is $R_{flex} = \frac{K_0+K_1+K_2}{N_1+N_2}$.

Inspired by the flexible length architecture, the local-global polar code architecture,  introduced in~\cite{Zhu2022}, connects multiple inner codes $G_1, ..., G_M$ through a systematic outer polar code. We assume these codes have the same length $N_i=N, i=1, \ldots, M$. 
A word of $K_b$ information bits is divided into $M$ parts of $K_{b_1}, \ldots, K_{b_M}$ bits that are assigned to the good bit-channels within the inner codes.  
The $K_a$ outer information bits are divided into $M$ parts of $K_{a_1}, \ldots, K_{a_M}$ bits that are mapped to the semipolarized bit-channels of the $M$ inner codes, respectively. The $P_a$ parity bits  of the outer codeword are similarly partitioned into $M$ parts of $P_{a_1}, \ldots, P_{a_M}$ bits and mapped to the remaining semipolarized bit-channels within the inner codes. This architecture supports local decoding of information bits $K_{a_i}, K_{b_i}$ within each  inner code $G_i$, with the option of improved decoding of the $M$ inner codewords via global decoding  using the outer code.    

\subsection{Stopping Set in Factor Graph}

We briefly review the stopping set analysis of polar codes as presented in~\cite{Eslami2013}, and we propose some new definitions that will be used throughout the rest of the paper.

\subsubsection{Notations from Eslami et al.~\cite{Eslami2013}}

\begin{figure}[htbp]
\centerline{\includegraphics[width=9cm,height=6.5cm]{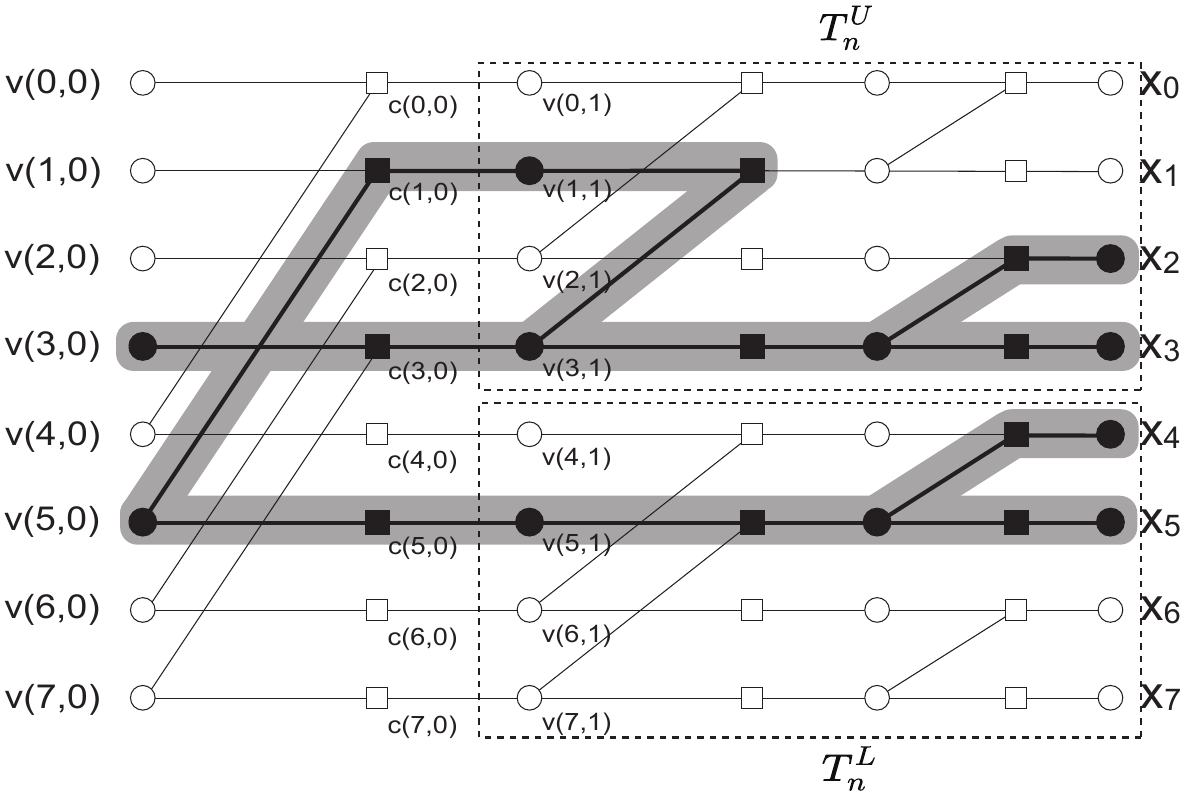}}
\caption{Normal realization of the factor graph for $N=8$. An example of a SS is shown with black variable and check nodes.}
\label{GSS}
\end{figure}

A stopping set (SS) is a non-empty set of variable nodes such that each neighboring check node is connected to this set at least twice. In this paper, we are particularly interested in the analysis of stopping sets in the factor graph of polar codes. Fig.~\ref{GSS} shows an example of a stopping set in the polar code factor graph, where we also included the corresponding set of check nodes. Denote the factor graph of a polar code of length $N = 2^n$ by $T_n$. A key observation is the symmetric structure of this graph which reflects the recursive construction of  the generator matrix: $T_{n}$ includes two factor graphs equivalent to $T_{n-1}$ as its upper and lower halves, connected together via $v(0,0),v(1,0),...,v(N-1,0)$ and $c(0,0),c(1,0),...,c(N-1,0)$. We denote these two subgraphs by $T_{n}^U$ and $T_{n}^L$, as shown in Fig.~\ref{GSS}. 

A stopping tree (ST) is a SS that contains one and only one information bit, i.e., variable node, on the leftmost stage of the sparse polar code factor graph.  For each information bit $i$, there is a unique stopping tree denoted by $ST(i)$. An example of such a stopping tree is shown in Fig.~\ref{ST} with black variable nodes. We also included the corresponding set of check nodes in order to visualize the structure of the tree. The size of the leaf set (variable nodes on the rightmost stage) of $ST(i)$ is denoted by $f(i)$. For example, $f(5) = 4$, with the corresponding leaf set $\{x_0,x_1,x_4,x_5\}$.

Only variable nodes on the rightmost stage are observed nodes, with all other variable nodes hidden. The set of observed variable nodes in a SS forms a variable-node SS (VSS). Accordingly, we define a minimum VSS (MVSS) as a VSS with a minimum number of observed variable nodes, among all the VSSs. The size of a MVSS is the stopping distance of the code. For any given index set $\mathcal J$, we denote a SS whose information nodes are precisely $\mathcal J$ as $SS(\mathcal J)$. The set of observed variable nodes in a $SS(\mathcal J)$ is a VSS for $\mathcal J$, denoted $VSS(\mathcal{J})$. A minimum size VSS among all the $VSS(\mathcal{J})$ is called a minimum VSS for $\mathcal J$, denoted $MVSS(\mathcal{J})$. Note that $SS(\mathcal J)$, $VSS(\mathcal J)$ and $MVSS(\mathcal J)$ may not be unique for a given index set $\mathcal{J}$. The following theorem is taken from~\cite{Eslami2013}, and we aim to extend it in Section IV.

\begin{thm}
\textbf{\textit{(Lower Bound I)}} Given any set $\mathcal{J}$ of information bits, we have $|MVSS(\mathcal{J})| \geq \min\limits_{j \in \mathcal{J}} f(j)$.
\label{Bound Eslami}
\end{thm}

\begin{proof}
The proof can be found in the Appendix of~\cite{Eslami2013}.
\end{proof}

Define $SD(\mathcal{A}) = \min\limits_{\mathcal{J} \subseteq \mathcal{A}} |MVSS(\mathcal{J})|$ as the stopping distance of a polar code  with information set $\mathcal{A}$. Theorem~\ref{Bound Eslami} sets a lower bound on the size of a MVSS for a set $\mathcal{J}$ of information bits. It also implies that the size of a MVSS for a polar code with information set $\mathcal{A}$ is at least equal to $\min\limits_{i \in \mathcal{A}} f(i)$. However, we already know that the leaf set of the stopping tree for any node $i \in \mathcal{A}$ is a VSS of size $f(i)$. The following corollary can be directly established: 

\begin{corollary}
    For a polar code with information bit index $\mathcal{A}$, $SD(\mathcal{A}) = \min\limits_{i \in \mathcal{A}} f(i)$.
\label{SD polar}
\end{corollary}

\begin{figure}[htbp]
\centerline{\includegraphics[width=9cm,height=5.5cm]{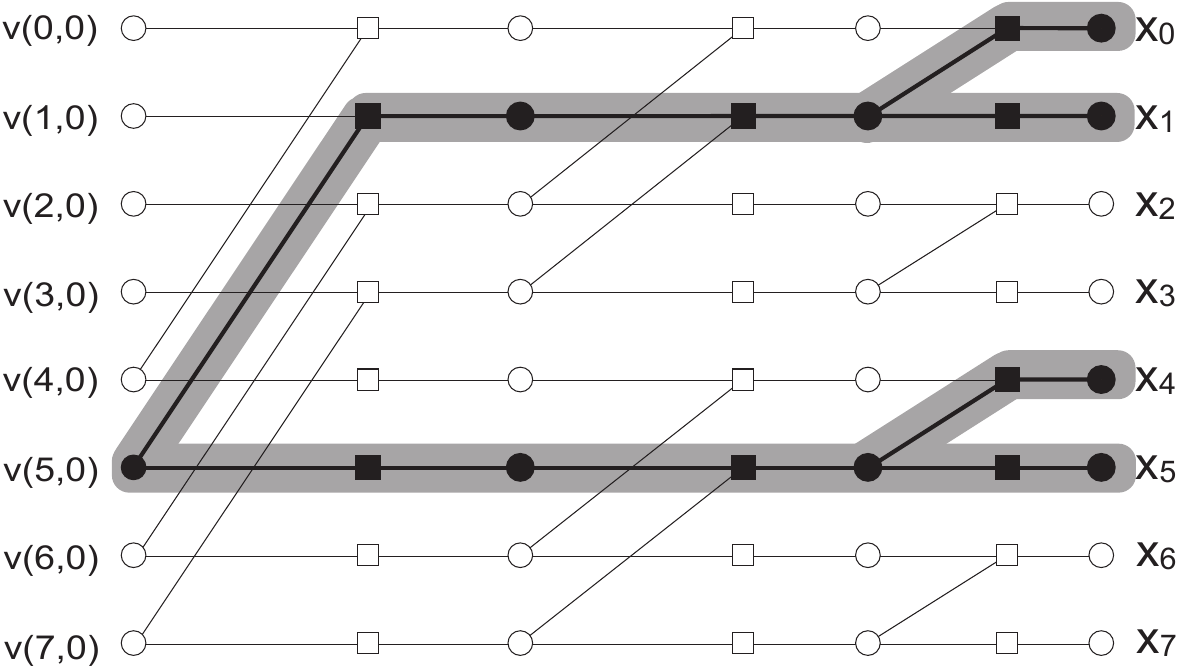}}
\caption{The stopping tree for v(5,0) is shown with black variable and check nodes.}
\label{ST}
\end{figure}

\subsubsection{New notations introduced in this paper}

Let $i \in \mathbb{Z}_N$ have the binary representation $i_b = i_0,i_1,...,i_{n-1}$, where $i = \sum _{k=0}^{k=n-1} i_{k} \times 2^{k}$, and let $wt(i_b)$ denote the weight of $i_b$ (i.e., the number of 1s). For example, if $i=6$, then $i_b=011$ and $wt(i_b)=2$. The following proposition provides a straightforward method to calculate $f(i)$. 

\begin{proposition}
    For length-$2^n$ polar codes, we have 
\begin{equation*}
f(i) = 2^{wt(i_b)}, 0\leq i \leq 2^n-1.
\end{equation*}
\label{weight_leaf}
\end{proposition}

\vspace{-20pt} 

\begin{proof}
From Fact 5 in \cite{Eslami2013}, we know that for a length-$2^n$ polar code, $f(i) = wt(r_i)$  for $0\leq i\leq 2^n-1$, where $r_i$ is the ($i+1$)-th row of $G =F^{\bigotimes n}$. To complete the proof, we need to show that $wt(r_i) = 2^{wt(i_b)}$. This can be done by induction. Denote the encoding matrix for a length-$2^n$ polar code as $G^n =F^{\bigotimes n}$, and denote the ($i+1$)-th row of $G^n =F^{\bigotimes n}$ as $r_i^n$.

For the case $n=1$, the statement follows immediately from inspection of the matrix $G^1=F$ and inspection of the corresponding factor graph for $n=1$. Now, suppose the result is true for a given $n$. Note that for a length-$2^{n+1}$ polar code, $G^{n+1}$ is recursively represented as 
\begin{equation}
    G^{n+1} = 
            \begin{bmatrix}
                G^n & 0 \\
                G^n & G^n 
            \end{bmatrix}.
\label{recursive_G}
\end{equation}
We see that, for $0\leq i\leq 2^n-1$, $r_i^{n+1} = [r_i^n,0,...,0]$. Thus, $f(i) = wt(r_i^{n+1}) = wt(r_i^n) = 2^{wt(i_b)}$. Similarly, for $2^n\leq i\leq 2^{n+1}-1$, we have $r_i^{n+1} = [r_{i-2^n}^n,r_{i-2^n}^n]$. Thus, $f(i) = wt(r_i^{n+1}) = 2 \times wt(r_{i-2^n}^n) = 2 \times 2^{wt(i_b) - 1} = 2^{wt(i_b)}$. This completes the induction. 
\end{proof}

Given information set $\mathcal{J}$, we use $UT(\mathcal{J}) = \cup _{j\in \mathcal{J}} ST(j)$ to denote the union of all the stopping trees defined by the elements in $\mathcal{J}$, which is also the largest $SS(\mathcal{J})$, with all other $SS(\mathcal{J}) \subseteq UT(\mathcal{J})$, as stated in the following proposition. Fig.~\ref{UT} gives an example of $UT(\mathcal{J})$ with $\mathcal{J} = \{3,5\}$. Define the degree of a node as the number of its neighboring nodes. Specifically, the degree of a variable node is the number of its neighboring check nodes, and the degree of a check node is the number of its neighboring variable nodes. In Fig.~\ref{UT}, $c(1,0)$ and $v(1,1)$ have a  degree of 2, while $c(1,1)$ and $v(3,1)$ have a degree of 3.

\begin{proposition}
    Any $SS(\mathcal{J})$ is a subset of $UT(\mathcal{J})$.
\label{subset}
\end{proposition}

\begin{proof}
    Assume there exists a node $v(r^*,c^*) \notin UT(\mathcal{J})$ such that $v(r^*,c^*)$ is in some $SS(\mathcal{J})$. If $c(r^*,c^*-1)$ is of degree~2 in the factor graph, then its neighboring variable node in column $c^*-1$ must be in $SS(\mathcal{J})$. Alternatively, if $c(r^*,c^*-1)$ is of degree~3, then at least one of its two neighboring variable nodes in column $c^*-1$ is in $SS(\mathcal{J})$. Without loss of generality, assume $v(r^*,c^*-1)$ is in $SS(\mathcal{J})$. Clearly, $v(r^*,c^*-1)$ is a parent of $v(r^*,c^*)$. We can then apply the same process to find the parent(s) in columns $c^*-2$, $c^*-3$, and so on, until we reach column $0$. 
    
    At this point, there must exist a node $v(r^{**},0)$ that is a parent of $v(r^*,c^*)$, implying $v(r^*,c^*) \in ST(r^{**})$ and that any $SS(\mathcal{J})$ containing $v(r^*,c^*)$ must also contain $v(r^{**},0)$. 
    However, since we assumed $v(r^*,c^*) \notin UT(\mathcal{J})$, it follows that $v(r^{**},0) \notin \mathcal{J}$, which contradicts the fact that 
    $v(r^{**},0)$ is in $SS(\mathcal{J})$.
\end{proof}

\vspace{-12pt} 

A degree-3 check node in $UT(\mathcal{J})$ must lie in the intersection of two stopping trees. We refer to this as an 
{\it intersection check node (ICN)}, and we denote the set of these as $ICN(\mathcal{J})$.
A leaf that is shared by more than one $ST(i)$, $i \in \mathcal{J}$ is referred to as an {\it overlapped leaf (OLL)}. The set of overlapped leaves is denoted as $OLL(\mathcal{J})$.
The leaves that are associated with exactly one $ST(i)$, $i \in \mathcal{J}$ are called 
{\it non-overlapped leaves} and denoted as $nOLL(\mathcal{J})$. 
For each element in $OLL(\mathcal{J})$, there exists at least one parent ICN. 
The parent ICN of the leaf indexed by $i$ with the largest column index is called the root ICN of $i$, denoted as $rICN(i)$.
Again taking Fig.~\ref{UT} as an example, $nOLL(\mathcal{J}) = \{x_2,x_3,x_4,x_5\}$, $OLL(\mathcal{J}) = \{x_0,x_1\}$ and $rICN(0) = rICN(1) = c(1,1)$.

\begin{figure}[htbp]
\centerline{\includegraphics[width=8cm,height=3.7cm]{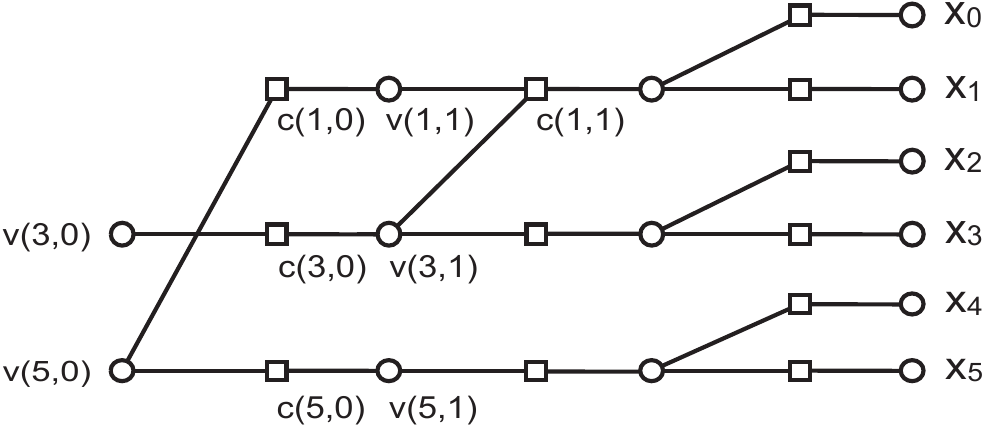}}
\caption{An example of $UT(\mathcal{J})$  for a set $\mathcal{J}$ of size 2.}
\label{UT}
\end{figure}

\section{Stopping set analysis for concatenated polar codes}

Consider the augmented polar code structure in Fig.~\ref{augmented_factor}, where an inner polar code is concatenated with an outer polar code. Recall that the inner information positions that are connected with the outer code are known as semipolarized channels, and the rest of the inner information bits are assigned to the good bit-channels~\cite{Elk2017}. To be specific, $K_0$ information bits are assigned to the auxiliary outer polar code, and an additional $K_1$ information bits are assigned to the good bit-channels within the inner code. The connection pattern between the outer codeword and the semipolarized bit-channels of the inner code is determined by an interleaver.

\begin{figure}[htbp]
\centerline{\includegraphics[width=9cm,height=4.8cm]{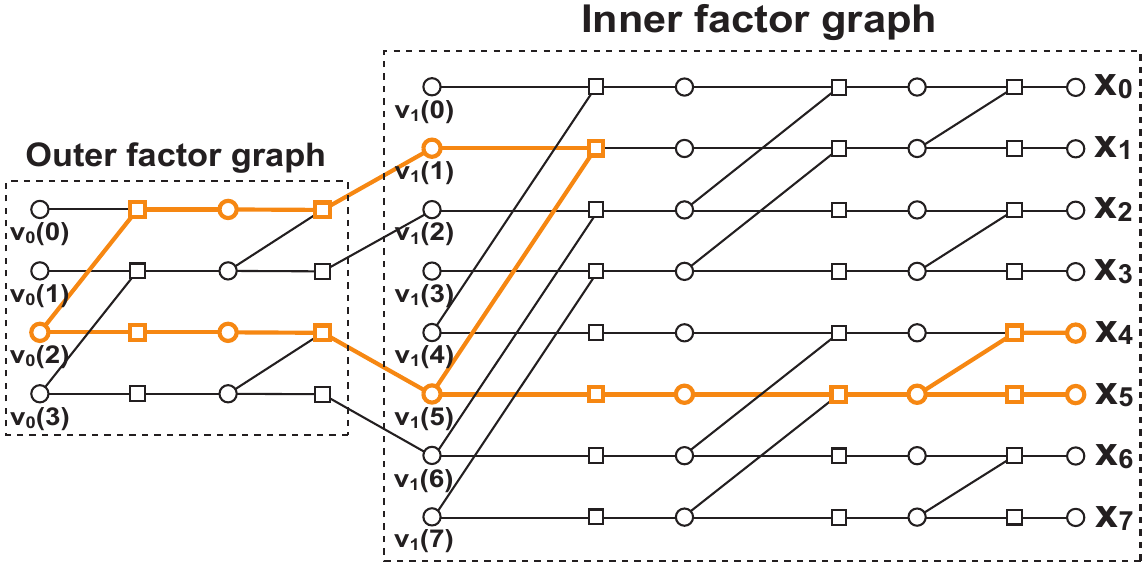}}
\caption{Augmented structure with $N_{0}{=}4$, $N_{1} {=}8$. Orange nodes represent a SS and $\{x_4,x_5\}$ are a MVSS.}
\label{augmented_factor}
\end{figure}

Let $v_{0}(i)$ denote the $(i+1)$-th node on the leftmost stage within the outer factor graph, and let $v_{1}(j)$ denote the $(j+1)$-th node on the leftmost stage within the inner factor graph. Let $\mathcal{H}_i$ denote the set of inner information nodes such that each element in $\mathcal{H}_i$ is connected to one of the leaves in the stopping tree $ST(i)$ defined on the outer factor graph. Let $MVSS(\mathcal H_i)$ be a MVSS (which may not be unique) defined by set $\mathcal H_i$ within the inner factor graph. For example in Fig.~\ref{augmented_factor}, $\mathcal H_2$ is the set of nodes $\{{v_{1}(1), v_{1}(5)}\}$. $MVSS(\mathcal H_2)$ is the set of nodes $\{{x_4,x_5}\}$. Note that $\{x_0,x_1,x_4,x_5\}$ is also a VSS for $\mathcal H_2$, but it is not a minimum VSS.

The information set $\mathcal{J}$ was previously defined as the set of indices of the information nodes. In this section, we slightly modify the definition of $\mathcal{J}$ to represent the information nodes themselves, rather than their indices, to avoid confusion between the elements in the set of outer information nodes and those in the set of inner information nodes. Denote the set of information nodes on the leftmost stage of the outer factor graph that corresponds to $K_0$ as $\mathcal{J}_{out}$, and denote the set of information nodes on the leftmost stage of the inner factor graph corresponding to $K_1$ as $\mathcal{J}_{in}$. Define $\mathcal{J}_{total} = \{\mathcal{J}_{out},\mathcal{J}_{in}\}$. 

Note that there is no valid stopping set within the overall factor graph whose intersection with the union of the leftmost stages of the inner and outer factor graphs is exactly $\mathcal{J}_{total}$. This is because, when $\mathcal{J}_{out}$ is non-empty, some inner information nodes corresponding to semipolarized channels must be included in the stopping set. However, $\mathcal{J}_{in}$ does not include any nodes corresponding to semipolarized channels. For example, in Fig.~\ref{augmented_factor}, with $\mathcal{J}_{out} = \{v_{0}(2)\}$ and $\mathcal{J}_{in} = \emptyset$, $\mathcal H_2$ must be part of any stopping set that only contains $\mathcal{J}_{out}$ on the leftmost stage. Therefore, we introduce the following definition: Let $SS'(\mathcal{J}_{total})$ denote a stopping set within the overall factor graph that includes exactly $\mathcal{J}_{out}$ on the leftmost stage of the outer factor graph, and $\mathcal{J}_{in}$ along with some semipolarized nodes, on the leftmost stage of the inner factor graph. 

\begin{example}
    We provide an example of SS' to demonstrate that at least one such set always exists for any given $\mathcal{J}_{total}$. Let $UT(\mathcal{J}_{out})$ represent the union of stopping trees in the outer factor graph, and let $\mathcal{H}_{\mathcal{J}_{out}}$ be the set of inner information nodes, where each element in $\mathcal{H}_{\mathcal{J}_{out}}$ is connected to one of the leaves in $UT(\mathcal{J}_{out})$. Define $\mathcal{J}^* = \mathcal{H}_{\mathcal{J}_{out}} \cup \mathcal{J}_{in}$, and let $UT(\mathcal{J}^*)$ be the union of stopping trees in the inner factor graph. Then, the set $UT(\mathcal{J}_{out}) \cup UT(\mathcal{J}^*)$ forms a valid $SS'(\mathcal{J}_{total})$.
\end{example}

In the augmented polar code structure, only the variable nodes on the rightmost stage of the inner factor graph are observed nodes. Accordingly, let $VSS'(\mathcal{J}_{total})$ denote a set of observed variable nodes in $SS'(\mathcal{J}_{total})$, and let $MVSS'(\mathcal{J}_{total})$ denote a minimum VSS among all the $VSS'(\mathcal{J}_{total})$. Define $SD'(\mathcal{J}_{total}) = \min\limits_{\mathcal{J} \subseteq \mathcal{J}_{total}} |MVSS'(\mathcal{J})|$ as the stopping distance of the augmented code. This nomenclature is appropriate because $SD'(\mathcal{J}_{total})$ characterizes the minimum number of errors at the receiver that will cause some unrecoverable errors in the BP decoder. 
We now derive an upper bound on this stopping distance.



We begin by considering a single information node. 
For an inner information node $i \in \mathcal{J}_{in}$, the leaves of $ST(i)$ are $MVSS'(i)$. Therefore, we continue to use $f(i)$ to characterize the value of $|MVSS'(i)|$ for $i \in \mathcal{J}_{in}$. For an outer information node $i \in \mathcal{J}_{out}$, it can be seen that $\mathcal{H}_i$ must be included in $MVSS'(i)$. Thus, we are interested in determining $|MVSS(\mathcal{H}_i)|$ within the inner factor graph. Clearly, $MVSS(\mathcal{H}_i)$ forms a $VSS'(i)$, thus providing an upper bound on $|MVSS'(i)|$. Let $f_{out}(i) = |MVSS(\mathcal{H}_i)|$. The following theorem then provides an upper bound on $SD'(\mathcal{J}_{total})$.

\begin{thm}
    Let $\min(a,b)$ be the function that returns the minimum value of $a$ and $b$. Then, we have
\begin{equation*}
     SD'(\mathcal{J}_{total}) \leq \min(\min\limits_{j \in \mathcal{J}_{out}} f_{out}(j),\min\limits_{j \in \mathcal{J}_{in}} f(j)).
\end{equation*}
\label{stopping distance}
\end{thm}

\begin{proof}
To start with, we point out that $SD'(\mathcal{J}_{total}) \leq \min\limits_{j \in \mathcal{J}_{total}} |MVSS'(j)|$. This inequality arises because it defines a search space that is contained in the search space considered by $SD'(\mathcal{J}_{total})$, which examines all possible combinations of nodes $\mathcal{J} \subseteq \mathcal{J}_{total}$ and selects the smallest $|MVSS'(\mathcal{J})|$. 
Next, we note that $\min\limits_{j \in \mathcal{J}_{total}} |MVSS'(j)| = \min(\min\limits_{j \in \mathcal{J}_{out}} |MVSS'(j)|,\min\limits_{j \in \mathcal{J}_{in}} f(j))$. Finally, as previously discussed, $f_{out}(j)$ provides an upper bound on $|MVSS'(j)|$ for $j \in\mathcal{J}_{out} $, i.e., $|MVSS'(j)| \leq f_{out}(j)$, thereby completing the proof.

\end{proof}

Given that $f(i)$ can be easily calculated by using Proposition~\ref{weight_leaf}, it is of interest to determine the value of $f_{out}(i) = |MVSS(\mathcal{H}_i)|$. Since the connection pattern between the inner code and the outer code is determined by   
an interleaver, knowing $f_{out}(i)$ would be sufficient if we can find $|MVSS(\mathcal{J})|$ for any selection of $\mathcal{J}$ within the inner factor graph. Theorem~\ref{Bound Eslami} gives a lower bound on this value. 
Empirical results, discussed in the next section, show that this bound is loose when $\mathcal{J}$ is randomly chosen (though we will prove that the bound becomes tight for some specific choices of $\mathcal{J}$). However, we will introduce four useful bounds on $|MVSS(\mathcal{J})|$, including one lower bound and three upper bounds.

\section{Bounds on $|MVSS(\mathcal{J})|$}

As discussed in the previous section, the problem of finding an upper bound on the stopping distance for the augmented code can be reduced to determining $|MVSS(\mathcal{J})|$ for randomly chosen sets $\mathcal{J}$ within the conventional polar factor graph. However, finding the exact value is challenging. In this section, we propose four different bounds on $|MVSS(\mathcal{J})|$. The proposed lower bound performs better than the one described in Theorem~\ref{Bound Eslami}, particularly when $\mathcal{J}$ is randomly chosen, as we will demonstrate through simulation results. Among the three upper bounds, the Encoding Bound has the lowest time complexity, though it does not perform as well as the others, which will be evident from experiment results. The Deletion Bounds I and II are similar but applicable to different scenarios, which we will discuss later in this section.

\subsection{Lower bound on $|MVSS(\mathcal{J})|$}

Let $G_{\mathcal J}$ denote the submatrix of the inner encoding matrix $G=F^{\bigotimes n}$ consisting of the rows that correspond to $\mathcal J$. Again taking Fig.~\ref{augmented_factor} as an example, the resulting $G_{\mathcal H_2}$ consists of the second and sixth rows of the inner encoding matrix $G=F^{\bigotimes 3}$. The following theorem, presented in  
\cite{Zhu2024}, gives a lower bound on the size of $MVSS(\mathcal J)$.

\begin{thm} \textbf{\textit{(Lower Bound II)}}
Given any information set $\mathcal{J}$, we have $|MVSS(\mathcal J)| \geq g(G_{\mathcal J})$, where 
\begin{equation}
g(A_{p\times q}) = \sum_{j=1}^{q} \delta(\sum_{i=1}^{p} a_{ij} - 1) 
\end{equation}
\begin{equation}
\delta(x) =
\begin{cases}
    1 & \text{if } x = 0, \\
    0 & \text{otherwise}.
\end{cases}
\end{equation}
\label{th1}
\end{thm}
\begin{proof}
See Appendix.
\end{proof}

In words, the function $g(\cdot)$ counts the number of columns in a matrix that have weight one. Thus, for any given information indices $\mathcal{J}$ we easily calculate the lower bound on the size of a $MVSS(\mathcal{J})$ by looking at the generator matrix of the polar code.

\subsection{Upper bounds on $|MVSS(\mathcal{J})|$}

Given a vector $v = [v_0, v_1, \dots, v_{n-1}] \in \mathbb{F}^n$, the \textit{support set} of $v$, denoted as \text{supp}($v$), is the set of indices where the elements of $v$ are non-zero. Formally, \text{supp}($v$) $= \{ i \in \{0, 1, \dots, n-1\} \mid v_i \neq 0 \}$. For example, for the vector $v = [1, 0, 0, 1, 1]$, the support set is \text{supp}($v$) = $\{0, 3, 4\}$.

\begin{thm}
\textbf{\textit{(Encoding Bound)}} Let $u$ be a length-$N$ binary vector, whose support set is $\mathcal{J}$. Let $x = uG$. Then, the nodes on the rightmost stage that are indexed by the support set of $x$ form a $VSS(\mathcal{J})$, and we have:
\begin{equation*}
    |MVSS(\mathcal{J})| \leq wt(x).
\end{equation*}
\label{Bound1}
\end{thm}

\begin{proof}
In the encoding process of polar codes, all of the variable nodes in the factor graph are set to either 0 or 1. Each check node has an even number of neighboring variable nodes with value 1, in order to satisfy the parity check requirement. If we initialize the variable nodes on the leftmost stage of the factor graph with $u$, i.e., set nodes in $\mathcal{J}$ to 1 and nodes in $\mathcal{J}^c$ to 0, and update for each variable node on the other stages (which is basically the encoding process), then the value-1 nodes on the rightmost stage will be indexed by the support set of $x$. 

We state that after the encoding process, all the variable nodes with value 1 form a stopping set. The reason is that if a check node is connected with a value-1 variable node, then it must connect to exactly two \mbox{value-1} variable nodes to satisfy the parity check equations. Thus, we can pick all the \mbox{value-1} variable nodes on the rightmost stage to form a variable stopping set of $\mathcal{J}$.
\end{proof}

Fig.~\ref{bound1} gives an example of Theorem~\ref{Bound1}, where $\mathcal{J}=\{0,3,7\}$. The value-1 variable nodes together with their neighboring check nodes are shown in black. In this example, $x = [1,0,0,0,1,1,1,1]$, and the corresponding $VSS(\mathcal{J}) = \{x_0,x_4,x_5,x_6,x_7\}$. 

\begin{figure}[htbp]
\centerline{\includegraphics[width=9cm,height=5.5cm]{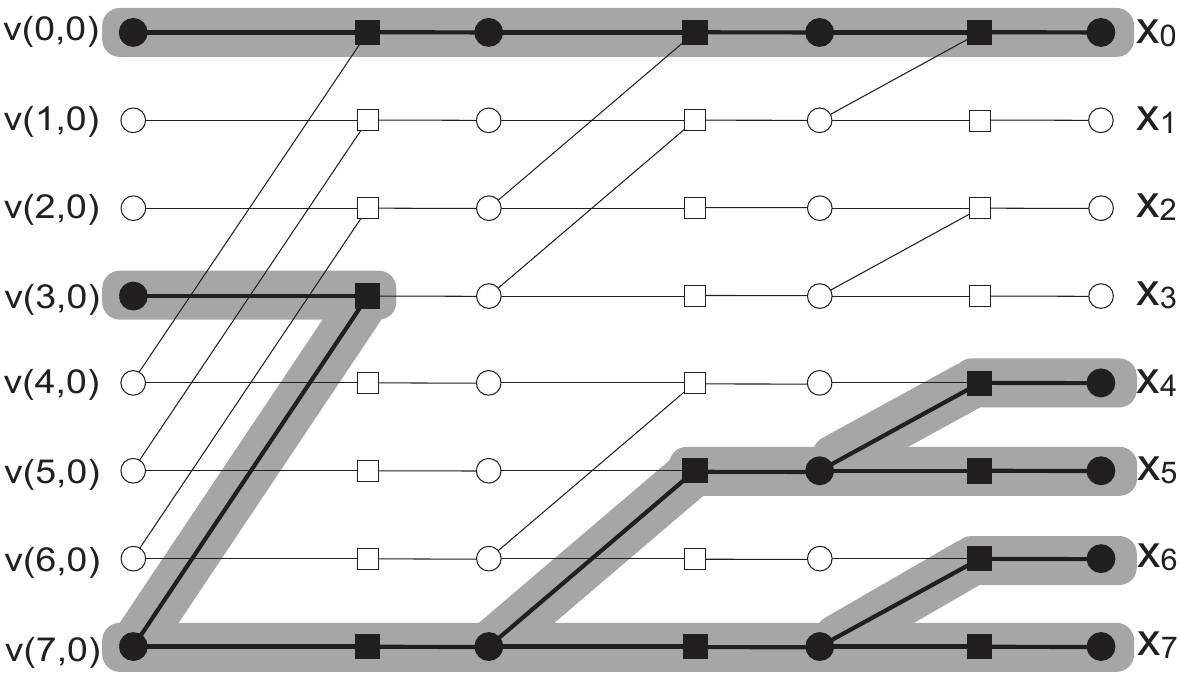}}
\caption{An example of the Encoding Bound, $\mathcal{J}=\{0,3,7\}$.}
\label{bound1}
\end{figure}

The following two upper bounds are algorithm-based. From Proposition~\ref{subset}, we know that any $MVSS(\mathcal{J})$ is a subset of $UT(\mathcal{J})$. To find a $MVSS(\mathcal{J})$, we begin by deleting some nodes from the set of leaf nodes in $UT(\mathcal{J})$. When the code length $N$ and $|\mathcal{J}|$ are small, such as when both are less than 32, we can perform an exhaustive search for $MVSS(\mathcal{J})$. For instance, in Fig.~\ref{UT}, the worst-case complexity involves testing all possible deletion patterns in the set of leaf nodes $\{x_0, x_1, x_2, x_3, x_4, x_5\}$, which amounts to $\sum_{k=1}^{k=5}{6\choose k} = 62$ combinations.

However, since the time complexity is not polynomial with respect to $|\mathcal{J}|$, a different approach is required for large $N$ and $|\mathcal{J}|$. Instead of randomly selecting combinations, we will need a certain deletion schedule. The following two upper bounds, derived from similar algorithms but employing  distinct schedules, are appropriate for different choices of $\mathcal{J}$. Before introducing the upper bounds, we reduce the search space for deletion patterns by observing that $nOLL(\mathcal{J})$ must be included in any $MVSS(\mathcal{J})$, as stated in the following proposition:

\begin{proposition}
    Given set $\mathcal{J}$ within the factor graph of a polar code, we have $nOLL(\mathcal{J}) \subseteq MVSS(\mathcal{J})$.
\label{nOLL}
\end{proposition}

\begin{proof}
    To see this, suppose that non-overlapped leaf node (nOLL) $x_k \in UT(\mathcal{J})$ belongs only to the stopping tree $ST(i)$, where $i\in \mathcal{J}$. The variable nodes on the branch of $ST(i)$ that traces back to the root node $v(i,0)$ must also be non-overlapped. This is because if one of the nodes $v(p,q)$ on that branch is shared by two trees or more, then all the children nodes of $v(p,q)$, i.e., nodes to the right of $v(p,q)$ along the tree $ST(i)$, must be shared nodes, including the leaf node $x_k$. This contradicts the assumption that $x_k$ is unshared. Since this branch belongs only to $ST(i)$, the result of deleting $x_k$ or any subset of nodes from this branch other than the root node $v(i,0)$ could not produce a stopping set, for this would mean that the remaining subset of nodes in $ST(i)$ would still constitute a stopping tree, call it $ST'(i)$, with root $v(i,0)$. However, this would violate Fact 2 in~\cite{Eslami2013}, which states that every information bit has a unique stopping tree. A simple example is shown in Fig.~\ref{bound1}, where $\mathcal{J}=\{0,3,7\}$ and $nOLL(\mathcal{J})=\{x_4,x_5,x_6,x_7\}$. There is no proper $SS(\mathcal{J})$ that does not include $nOLL(\mathcal{J})$.
\end{proof}

\begin{thm}
\textbf{\textit{(Deletion Bound I)}} Let $S$ be the set of variable nodes returned by Algorithm~\ref{algo_version1}. Then $S$ forms a $VSS(\mathcal{J})$, and we have:
\begin{equation*}
    |MVSS(\mathcal{J})| \leq |S|
\end{equation*}
\label{Bound2}
\end{thm}

\begin{proof}
From Proposition~\ref{nOLL}, we know that $MVSS(\mathcal{J})$ must include $nOLL(\mathcal{J})$. The algorithm attempts to delete some nodes from the set $OLL(\mathcal{J})$ by checking whether the punctured $UT(\mathcal{J})$ can still form a stopping set, i.e., by verifying if any degree-1 check nodes remain (see lines 9-11 in the algorithm). If $S$ is returned by the algorithm, then there exists a punctured subgraph containing $\mathcal{J}$ on the left and $S$ on the right, with some nodes in the middle, which forms a stopping set. 
\end{proof}

\begin{algorithm}
\caption{Find small VSS (Deletion Bound I)}\label{alg:cap}
\begin{algorithmic}[1]

\Statex \textbf{Input:} $\mathcal{J}$, $N$ ($N$ is used to initialize the factor graph)
\Statex \textbf{Output:} $VSS(\mathcal{J})$
\State find $OLL(\mathcal{J})$ and $nOLL(\mathcal{J})$
\State UT $\gets UT(\mathcal{J})$ 
\State $VSS(\mathcal{J}) \gets nOLL(\mathcal{J})$
\State OLL\_temp $\gets OLL(\mathcal{J})$
\While{OLL\_temp is not empty}
    \State UT\_punc = UT
    \State pick $l$ with the largest index from OLL\_temp
    \State delete children leaves of $rICN(l)$ from UT\_punc
    \While{Exist a degree-1 CN in UT\_punc}
        \State delete this degree-1 CN and its neighbor VN from UT\_punc
    \EndWhile
    \If{any VN on the leftmost stage is deleted}
        \State $VSS(\mathcal{J}) = VSS(\mathcal{J}) \cup \{l\}$
        \State remove $l$ from OLL\_temp
    \Else
        \State remove all the leaves deleted in this iteration from OLL\_temp
        \State UT = UT\_punc
    \EndIf
\EndWhile

\end{algorithmic}
\label{algo_version1}
\end{algorithm}

An example of Algorithm~\ref{algo_version1} is shown in Fig.~\ref{counter_example} with $\mathcal{J}=\{0,3,7\}$ and $N=8$. $UT(\mathcal{J})$ is shown by the colored nodes (black, green and orange). $nOLL(\mathcal{J}) = \{x_4,x_5,x_6,x_7\}$ and $OLL(\mathcal{J}) = \{x_0,x_1,x_2,x_3\}$. Algorithm~\ref{algo_version1} first tries to delete $x_3$ from $OLL(\mathcal{J})$ by trying to delete all the children leaf nodes of $c(3,0)$, which is $rICN(3)$, the rICN of the selected leaf $x_3$. However, deleting $\{x_0,x_1,x_2,x_3\}$ would result in deletion of $v(0,0)$, so the algorithm labels $x_3$ as undeletable. The same process applies to $x_2$ and $x_1$. 
The algorithm then moves to $x_0$. 
Deleting the only child leaf of $rICN(0)$ yields a structure that 
forms a stopping set that contains $\mathcal{J}$.
As a result, $x_0$ will be deleted from UT, and the punctured UT is stored for future iterations until all elements in $OLL(\mathcal{J})$ are tried. In this example, since $x_0$ is the last element, Algorithm~\ref{algo_version1} will terminate and return the set $\{x_1,x_2,x_3,x_4,x_5,x_6,x_7\}$.

\begin{figure}[htbp]
\centerline{\includegraphics[width=9cm,height=5.5cm]{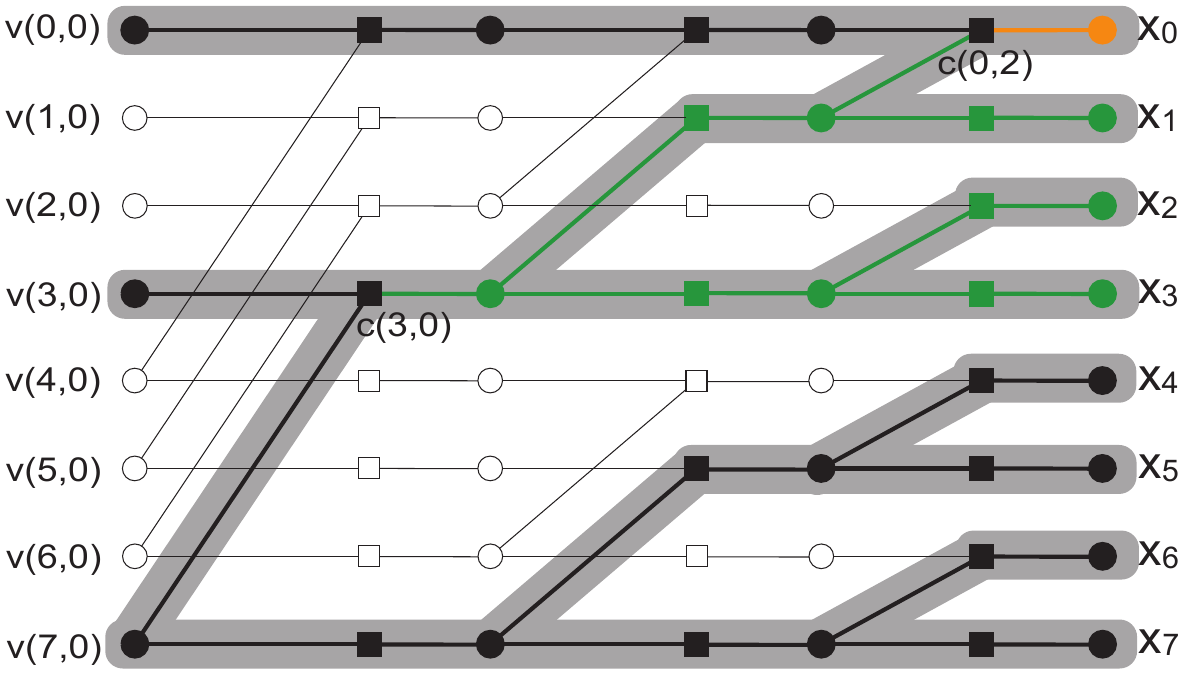}}
\caption{An example of Deletion Bounds I and II, $\mathcal{J}=\{0,3,7\}$.}
\label{counter_example}
\end{figure}

In the same example, the exact value of $|MVSS(\mathcal{J})|$ is 5, which is obtained by  deleting $\{x_1,x_2,x_3\}$. To improve the performance of Deletion Bound I, we modified Algorithm~\ref{algo_version1} as follows: Given a leaf node $l$, instead of identifying $rICN(l)$ and deleting all its children leaves, the new algorithm attempts to delete only the leaf node $l$ itself. Moreover, $l$ is randomly selected instead of being picked according to the largest index. This modified approach is detailed in Algorithm~\ref{algo_version2}. It can be observed that with sufficient attempts using different random deletion seeds, Algorithm~\ref{algo_version2} can find a $MVSS(\mathcal{J})$, as it effectively performs an exhaustive search. While Algorithm~\ref{algo_version2} appears to be more adaptive and precise, we will demonstrate through experimental results in Section V that Algorithm~\ref{algo_version1} performs better for certain choices of $\mathcal{J}$, such as when $\mathcal{J}$ is the information set of a polar codes or a Reed-Muller code.

\begin{thm}
\textbf{\textit{(Deletion Bound II)}} Let $S$ be the set of variable nodes returned by Algorithm~\ref{algo_version2}. Then $S$ forms a $VSS(\mathcal{J})$, and we have:
\begin{equation*}
    |MVSS(\mathcal{J})| \leq |S|
\end{equation*}
\label{Bound3}
\end{thm}

\begin{proof}
    The same proof as Theorem~\ref{Bound2}. 
\end{proof}

\begin{algorithm}
\caption{Find small VSS (Deletion Bound II)}\label{alg:cap}
\begin{algorithmic}[1]

\Statex \textbf{Input:} $\mathcal{J}$, $N$ ($N$ is used to initialize the factor graph)
\Statex \textbf{Output:} $VSS(\mathcal{J})$
\State find $OLL(\mathcal{J})$ and $nOLL(\mathcal{J})$
\State UT $\gets UT(\mathcal{J})$ 
\State $VSS(\mathcal{J}) \gets nOLL(\mathcal{J})$
\State OLL\_temp $\gets OLL(\mathcal{J})$
\While{OLL\_temp is not empty}
    \State UT\_punc = UT
    \State randomly pick $l$ from OLL\_temp
    \State delete $l$ from UT\_punc
    \While{Exist a degree-1 CN in UT\_punc}
        \State delete this degree-1 CN and its neighbor VN from UT\_punc
    \EndWhile
    \If{any VN on the leftmost stage is deleted}
        \State $VSS(\mathcal{J}) = VSS(\mathcal{J}) \cup \{l\}$
        \State remove $l$ from OLL\_temp
    \Else
        \State remove all the leaves deleted in this iteration from OLL\_temp
        \State UT = UT\_punc
    \EndIf
\EndWhile

\end{algorithmic}
\label{algo_version2}
\end{algorithm}


Consider again the example in Fig.~\ref{counter_example}. The output of Algorithm~\ref{algo_version2} has two possible outcomes. If the algorithm selects $l = x_1$ (or $x_2,x_3$), it will return $\{x_0,x_4,x_5,x_6,x_7\}$ with green nodes deleted. If it selects $l = x_0$, it will return $\{x_1,x_2,x_3,x_4,x_5,x_6,x_7\}$ with the orange node deleted. Thus, with a $75\%$ probability, the algorithm provides a bound of 5, and with a $25\%$ probability, it provides a bound of 7. 

In practice, we can run the algorithm $t$ times with different random seeds. By repeating the algorithm with varied seeds, we can gather a range of potential outcomes and subsequently select the smallest result as the final upper bound.

\subsection{Simulation results}

In this part, we present experiment results that demonstrate the performance of the proposed bounds. The code length for Figs.~\ref{Polar1024}-\ref{Random1024_another} is set to $N=1024$. The parameter $K$ represents the size of the set $\mathcal{J}$, which corresponds to the number of information bits on the leftmost stage of the factor graph. In Fig.~\ref{Polar1024}, $\mathcal{J}$ is selected to form polar codes designed using Bhattacharyya parameters \cite{Arikan2009}. In contrast, for Figs.~\ref{Random1024} and \ref{Random1024_another}, the set $\mathcal{J}$ is chosen randomly.

In Fig.~\ref{Polar1024}, Deletion Bound I completely aligns with Lower Bound I for all values of $K$, indicating that these two bounds can accurately determine $|MVSS(\mathcal{J})|$ when $\mathcal{J}$ is chosen to form a polar code. In the following section, we will prove that Lower Bound I can precisely determine the value when $\mathcal{J}$ satisfies certain conditions. Additionally, Fig.~\ref{Polar1024} shows that Encoding Bound is relatively loose, and while Deletion Bound~II does not always yield the exact value of $|MVSS(\mathcal{J})|$, it still provides a fairly good approximation. We ran Algorithm~\ref{algo_version2} once ($t=1$) to generate the curve for Deletion Bound II. Lower Bound II is quite loose, as it frequently yields values close to zero.

Fig.~\ref{Random1024} compares the proposed bounds when $\mathcal{J}$ is randomly selected. It can be observed that while Deletion Bounds I and II may yield different results, their outputs are very close to one another. Encoding Bound, however, is generally looser than both Deletion Bounds I and II. The performance of the proposed lower bounds differs from Fig.~\ref{Polar1024}: In Fig.~\ref{Polar1024}, where $\mathcal{J}$ is selected to form polar codes, Lower Bound II is loose, whereas in Fig.~\ref{Random1024}, with randomly chosen $\mathcal{J}$, it is Lower Bound I that is loose and often yields values close to zero.

Fig.~\ref{Random1024_another} compares the performance of Deletion Bound II with different values of $t$, alongside Deletion Bound I when $\mathcal{J}$ is randomly selected. It can be observed that when the code length $N$ is large, using a small value of $t$ (relative to the size of the search space) has minimal impact on the output of Algorithm~\ref{algo_version2}. Deletion Bounds I and II with a small $t$ produce nearly the same results, suggesting that a one-shot search may suffice for approximating the upper bound of $|MVSS(\mathcal{J})|$.

In Fig.~\ref{Exhaust32}, with $N=32$ and a randomly chosen $\mathcal{J}$, the green dashed line (which coincides with the red line) labeled ``$|MVSS|$'' represents the exact value of $|MVSS(\mathcal{J})|$, obtained through exhaustive search. It can be observed that Deletion Bound~II with $t=10$ accurately identifies the exact value of $|MVSS(\mathcal{J})|$, while Deletion Bound~I is occasionally less tight.  

Although the results may differ when Algorithm~\ref{algo_version2} is run with different random seeds, we emphasize that the outcomes are `robust' in the sense that they exhibit minimal variation across different random deletion seeds. 
This is confirmed in Fig.~\ref{Box}, which shows results for  $N=1024$, $K=256,512,768$,  with $\mathcal{J}$ selected to form a polar code.
In these experiments, Algorithm~\ref{algo_version2} was run 100 times with different random seeds for each value of $K$. Instead of picking only the smallest outcome from the 100 trials, we plot all of them in the form of a box plot, with middle half of the results falling within the colored box. 
The results show little fluctuation. 
\vfill

\begin{figure}[htbp]
\centerline{\includegraphics[width=8cm,height=6cm]{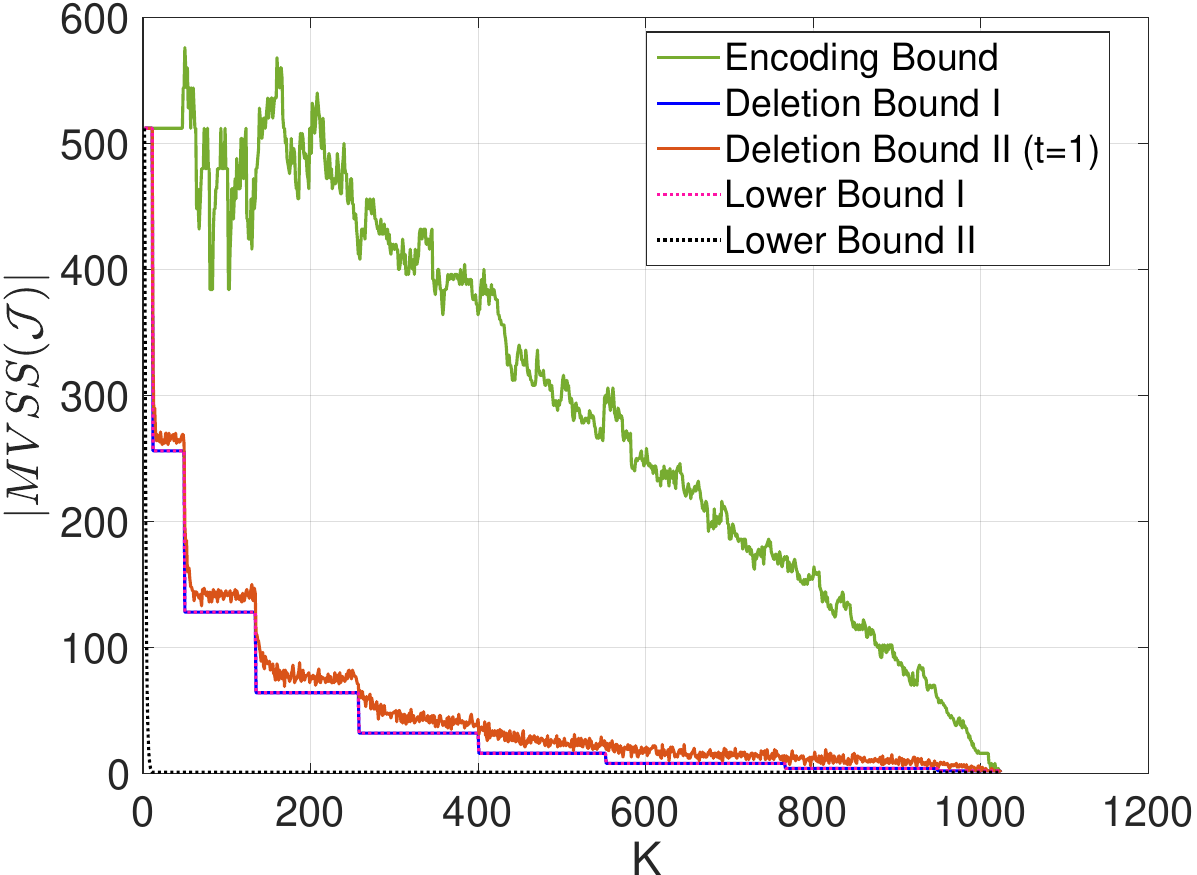}}
\caption{A comparison of different bounds with $N=1024$ and $\mathcal{J}$ selected to form polar codes.}
\label{Polar1024}
\end{figure}

\begin{figure}[htbp]
\centerline{\includegraphics[width=8cm,height=6cm]{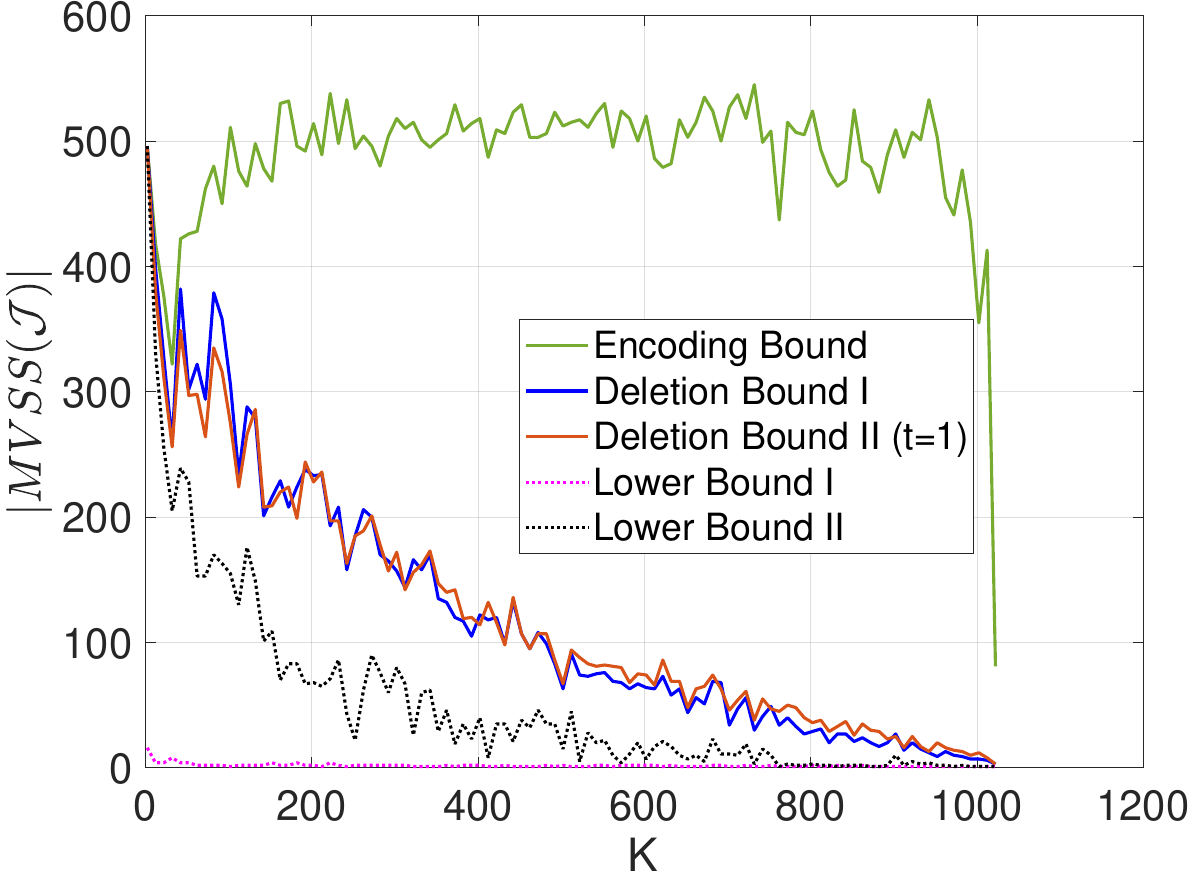}}
\caption{A comparison of different bounds with $N=1024$ and randomly selected $\mathcal{J}$.}
\label{Random1024}
\end{figure}

\begin{figure}[htbp]
\centerline{\includegraphics[width=8cm,height=6cm]{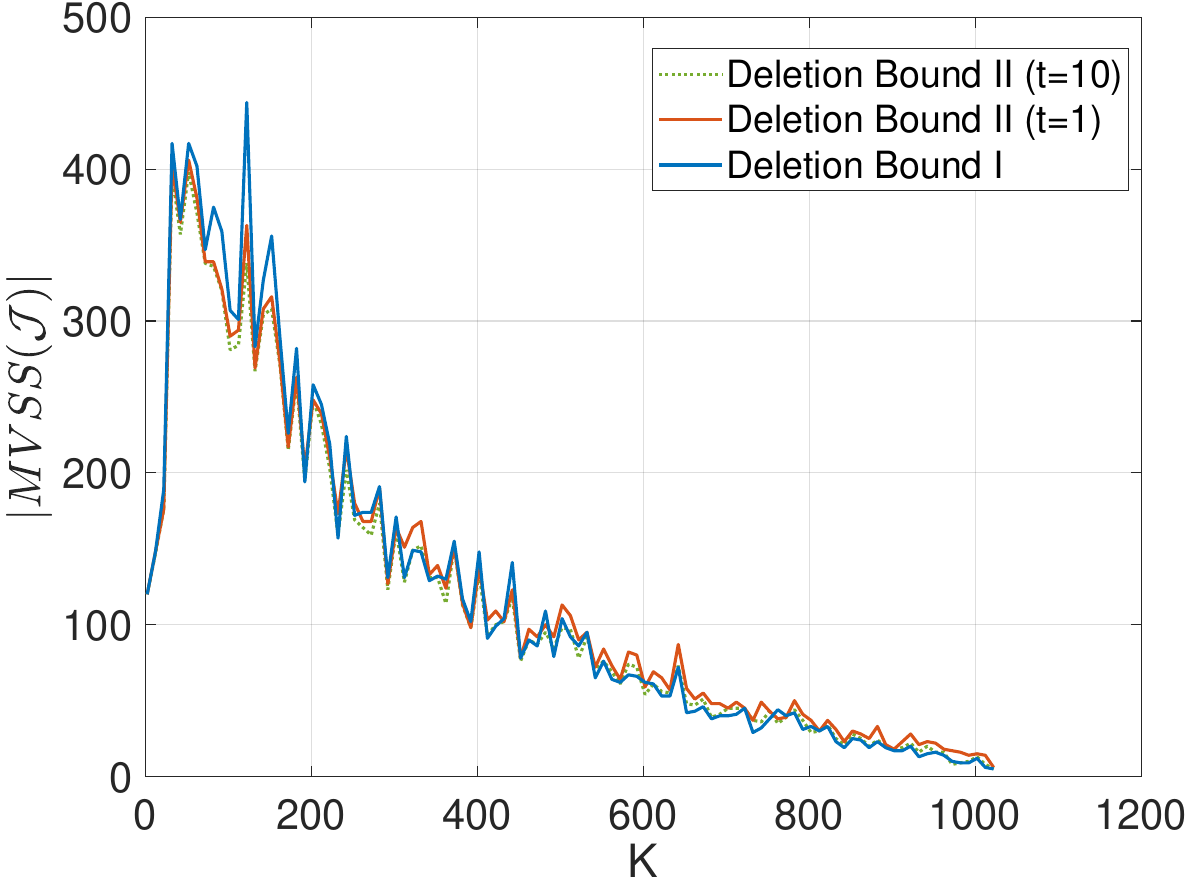}}
\caption{A comparison of Deletion Bounds I and II with $N=1024$ and randomly selected $\mathcal{J}$.}
\label{Random1024_another}
\end{figure}

\begin{figure}[htbp]
\centerline{\includegraphics[width=8cm,height=6cm]{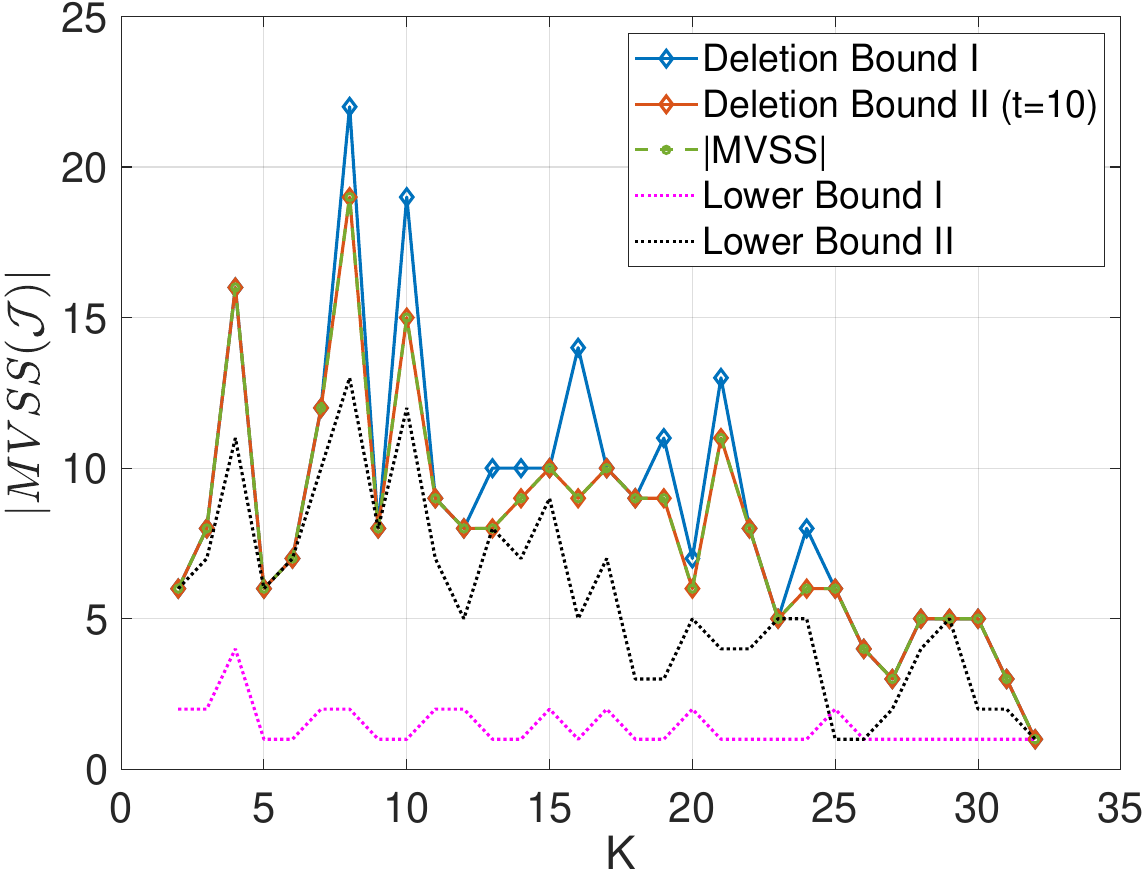}}
\caption{A comparison of different bounds with $N=32$ and randomly selected $\mathcal{J}$, alongside the exact value of $|MVSS(\mathcal{J})|$.}
\label{Exhaust32}
\end{figure}

\begin{figure}[h]
\centerline{\includegraphics[width=8cm,height=6cm]{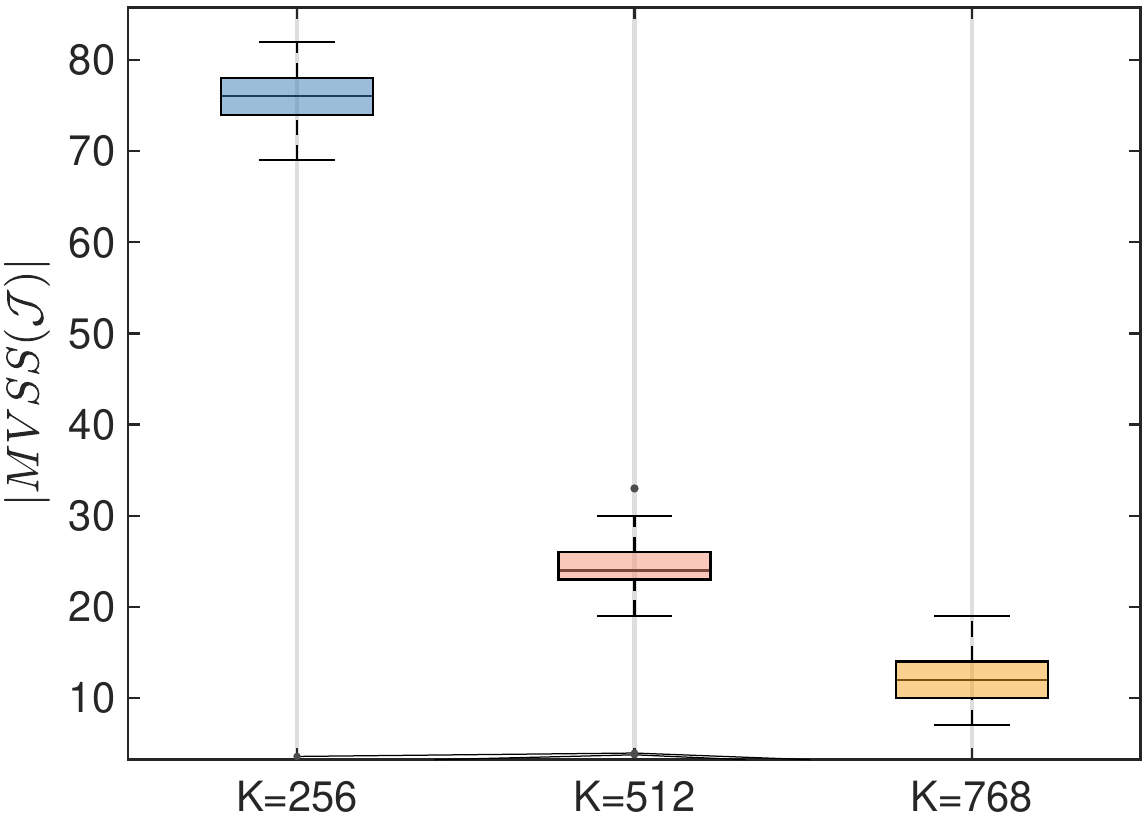}}
\caption{Box plots of results from Algorithm~\ref{algo_version2} with different random seeds for $N=1024$, $K=256,512,768$, and $\mathcal{J}$   selected to form a polar code.}
\label{Box}
\end{figure}

\vfill
\section{$|MVSS(\mathcal{J})|$ for specific choices of $\mathcal{J}$}

\subsection{Case 1}

The following result shows that the bound described in Theorem~\ref{th1} is tight when $|\mathcal{J}|=2$.

\begin{thm}
When $|\mathcal{J}|=2$, $|MVSS(\mathcal J)| = g(G_{\mathcal J})$.
\label{a1}
\end{thm}

\textbf{Proof:} Let $\mathcal{J} = \{i,j\}$. Then $UT(\mathcal{J}) = ST(i) \cup ST(j)$. Assume $x_k \in UT(\mathcal{J})$ is an overlapped leaf (OLL) belonging to both $ST(i)$ and $ST(j)$.  Then there must exist a degree-3 check node in the graph of $UT(\mathcal{J})$, with all three neighboring variable nodes in $UT(\mathcal{J})$, that lies in the intersection of the graphs of  $ST(i)$ and $ST(j)$ and has $x_k$ as a child node. Otherwise, one could trace back from $x_k$ to a single root node, contradicting the fact that it is an OLL.  All of the children nodes to the right of that degree-3 check node must be shared by the two trees. This implies that we can delete those children nodes, which include the shared node $x_k$, and the remaining structure will still be a stopping set in $UT(\mathcal{J})$. This procedure can be repeated for any remaining OLLs, until all of the original OLLs in $UT(\mathcal{J})$ have been deleted. The only remaining leaf nodes are the original nOLLs.

Recall from the proof of Theorem~\ref{th1} that the indices of the leaf nodes in the stopping tree $ST(i)$ are given by the positions of the ones in $r_i^n$, where $r_i^n$ is the $(i+1)$-th row of $G^n = F^{\bigotimes n}$, the encoding matrix for polar codes of length $2^n$. Consequently, the columns of $G_{\mathcal J}$ with weight one correspond to the nOLLs, and $g(G_{\mathcal J})$ returns the number of such columns, thereby completing the proof.
\qed

Fig.~\ref{MVSS} illustrates the proof procedure, where $\mathcal{J} = \{2,6\}$, and $UT(\mathcal{J}) = ST(2) \cup ST(6)$. The OLLs are $x_0$ and $x_2$, while the nOLLs are $x_4$ and $x_6$. The green nodes in $UT(\mathcal{J})$ are the children of the orange degree-3 check node that lies in the intersection of the graphs of $ST(2)$ and $ST(6)$. Both $x_0$ and $x_2$ are children variable nodes of this check node. The orange nodes represents the stopping set corresponding to $MVSS(\mathcal{J})$ that remains after deleting the green nodes. 

\begin{figure}[htbp]
\centerline{\includegraphics[width=9cm,height=5.3cm]{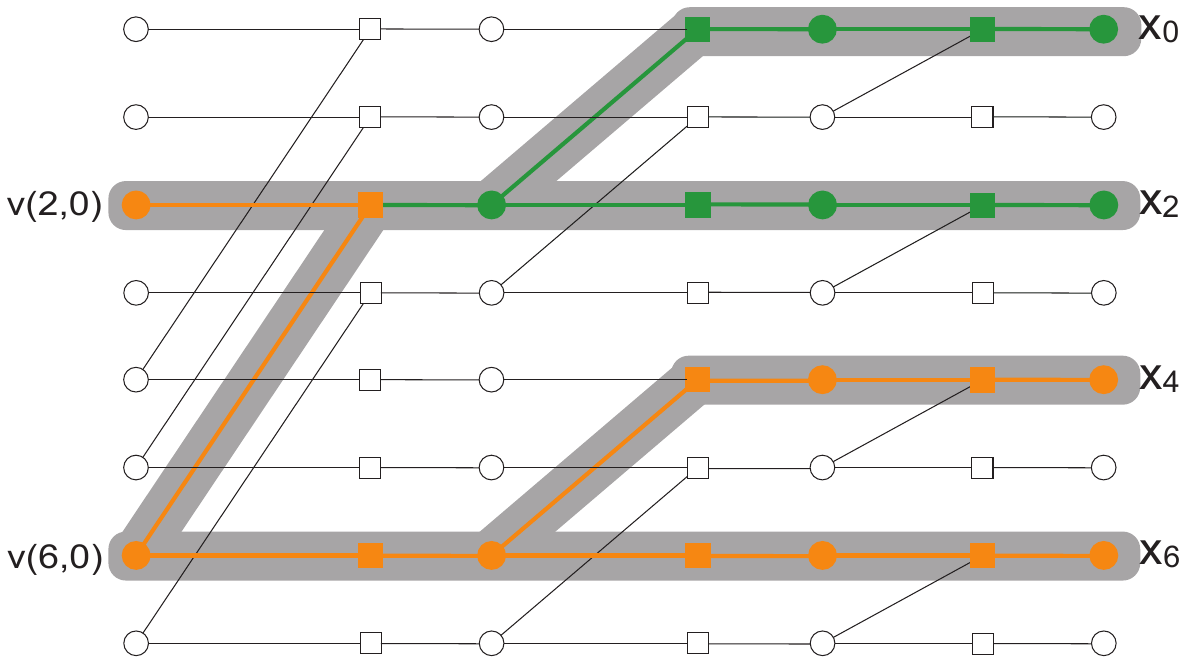}}
\caption{An illustration of Theorem~\ref{a1}.}
\label{MVSS}
\end{figure}

We note that Theorem~\ref{a1} cannot be extended to the case where $|\mathcal{J}| > 2$. For example in Fig.~\ref{MVSS_a}, if $\mathcal{J} = \{1,6,7\}$, then the nodes $\{x_0,x_1,x_2,x_4,x_6\}$ are shared leaves. There are two MVSSs for $\mathcal{J}$: $\{x_0, x_3, x_5, x_7\}$ and $\{x_1, x_3, x_5, x_7\}$. Removing the OLLs for $\mathcal{J}$ does not yield an $MVSS(\mathcal{J})$; rather, the size of the resulting VSS provides a lower bound for $|MVSS(\mathcal{J})|$.

\begin{figure}[htbp] 
	\centering  
	\vspace{-0.35cm} 
	\subfigtopskip=5pt 
	\subfigbottomskip=5pt 
	\subfigcapskip=0pt 
	\subfigure[The SS corresponding to $MVSS(\mathcal{J}) = \{x_0, x_3, x_5, x_7\}$.]{
		\includegraphics[width=7.5cm,height=4.4cm]{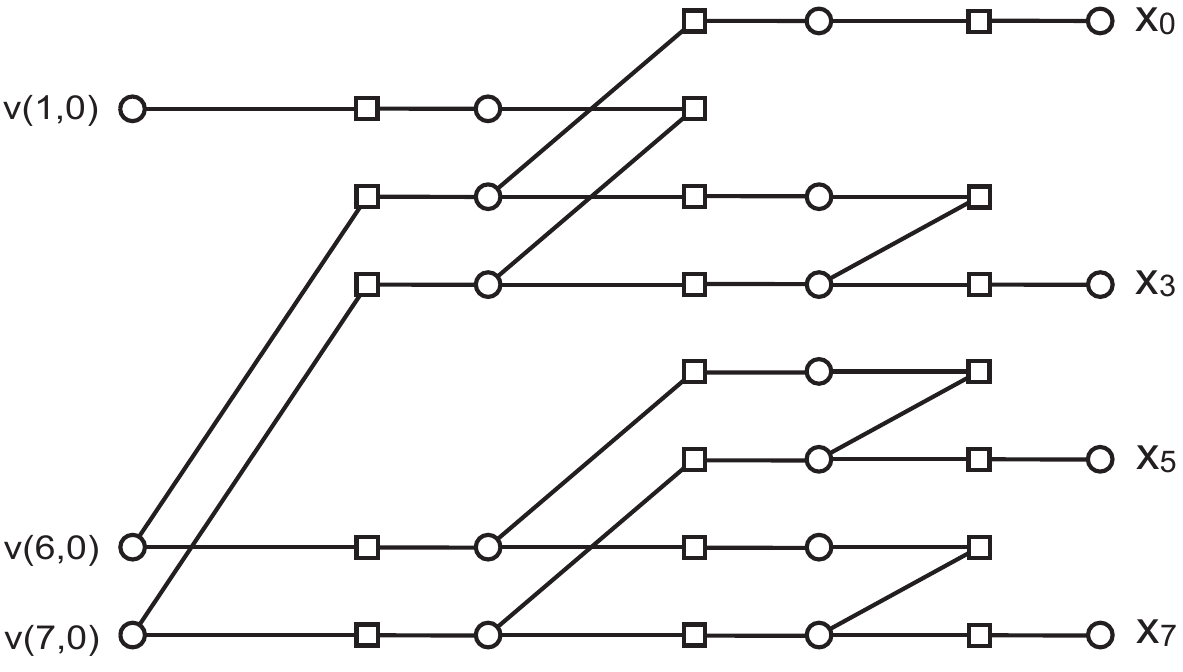}}
	\quad 
	\subfigure[The SS corresponding to $MVSS(\mathcal{J}) = \{x_1, x_3, x_5, x_7\}$.]{
		\includegraphics[width=7.5cm,height=4.4cm]{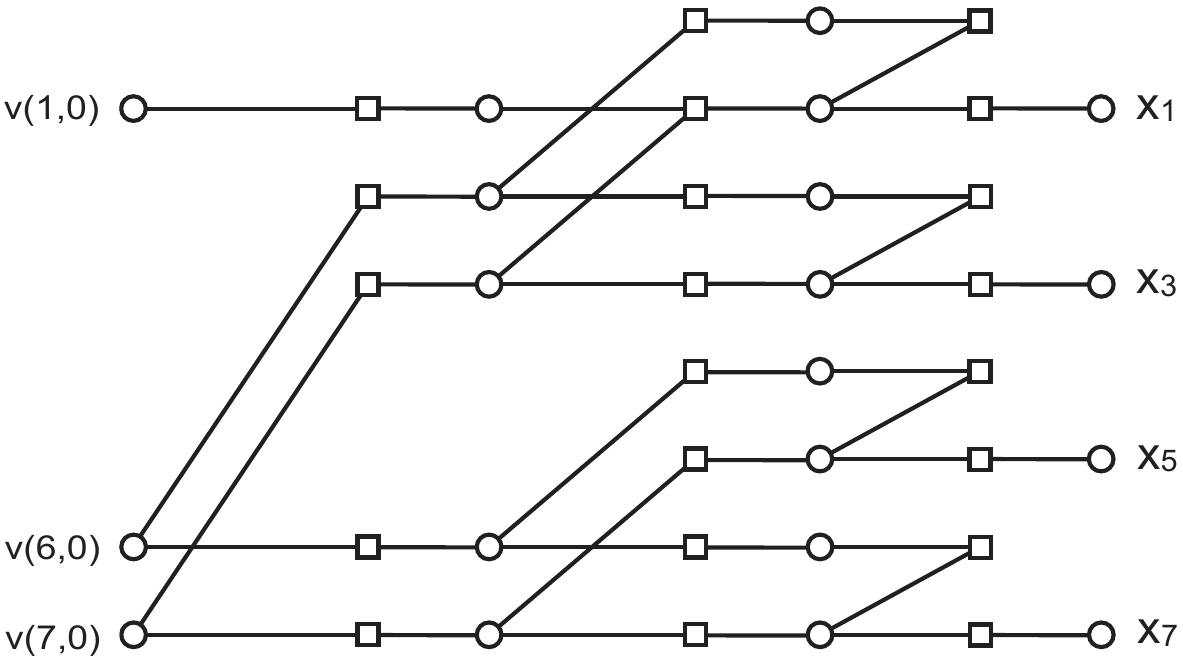}}
	  \\
        \caption{Example showing that Theorem~\ref{a1} does not generalize to $|\mathcal{J}|>2$.}
\label{MVSS_a}        
\end{figure}

\subsection{Case 2}

Here we give another condition on $\mathcal{J}$, under which the value of $|MVSS(\mathcal{J})|$ can be determined in $O(N)$, where $N$ is the code length. In words, given information set $\mathcal{J}$ and $f(i)$ for each $i \in \mathcal{J}$, $|MVSS(\mathcal{J})|$ can be calculated by selecting the smallest $f$ value. We start with some definitions that are used in the rest of this subsection.

For any given information set $\mathcal{J}$, there always exists an information bit $j\in\mathcal{J}$ whose corresponding stopping tree has the smallest leaf set among all the elements in $\mathcal{J}$. We call such an information bit a \textit{minimum information bit} for $\mathcal{J}$, denoted by $MIB(\mathcal{J})$. Note that there may exist more than one MIB in $\mathcal{J}$. In that case, we pick the one with the largest index. We denote the selected MIB as $MIB^*(\mathcal{J})$. That is, we pick the one which occupies the lowest place in the graph among the MIBs of $\mathcal{J}$.

Let $\mathcal{J}_n$ denote the set of indices on the leftmost stage of the factor graph $T_n$, and let $\mathcal{J}_{n}^L$ denote the set of indices on the leftmost stage of $T^L_{n}$, such that the variable nodes indexed by $\mathcal{J}_{n}^L$ on the leftmost stage of $T^L_{n}$ are in $UT(\mathcal{J}_n)$. Similarly, $\mathcal{J}_{n}^U$ denotes the set of indices whose corresponding nodes on the leftmost stage of $T^U_{n}$ are in $UT(\mathcal{J}_n)$. For example in Fig.~\ref{UT}, $\mathcal{J}_n = \{ 3,5\}$ (with corresponding nodes $v(3,0),v(5,0)$), $\mathcal{J}_n^L = \{1\}$ (with corresponding node $v(5,1)$) and $\mathcal{J}_n^U = \{1,3\}$ (with corresponding nodes $v(1,1),v(3,1)$).

Recall that $i,j \in \mathbb{Z}_N$ have binary representations $i_b = i_0,i_1,...,i_{n-1}$ and $j_b = j_0,j_1,...,j_{n-1}$, i.e., $i = \sum _{k=0}^{k=n-1} i_{k} \times 2^{k}$ and $j = \sum _{k=0}^{k=n-1} j_{k} \times 2^{k}$, respectively. 

\begin{defn}
    We write $j \nearrow i$ if there exists $k,k' \in \mathbb{Z}_n$ with $k<k'$ such that

    1) $j_k=1$ and $j_{k'}=0$
    
    2) $i_k=0$ and $i_{k'}=1$
    
    3) For all $l \in \mathbb{Z}_n \setminus \{k,k'\}: j_l=i_l$
\end{defn}

Clearly, if $j \nearrow i$ then $j<i$ and $wt(j_b) = wt(i_b)$. 

\begin{defn}
    Define the following conditions as:

    \textbf{(cover condition)} (Definition 7 from \cite{PO2016}.) If $j \in \mathcal{J}$ and for all $k \in \mathbb{Z}_n$, we have $j_k=1 \Rightarrow i_k=1$ (meaning that $i$ covers $j$), then $i \in \mathcal{J}$. 

    \textbf{(swap condition)} (Definition 4 from \cite{PO2016}.) If $j \in \mathcal{J}$ and $j \nearrow i$, then $i \in \mathcal{J}$. 
\end{defn}
\label{defn_conditions}

In~\cite{PO2016}, it is also stated that if $i$ covers $j$, or $j \nearrow i$, then $W_N^{(j)}$ is stochastically degraded with respect to $W_N^{(i)}$.

Given a set of indices $\mathcal{J}_n$, define $\overline{\mathcal{J}_{n}} \triangleq \{ j\in \mathcal{J}_n | j < 2^{n-1}\}$ and define $\underline{\mathcal{J}_{n}} \triangleq \{ j - 2^{n-1} | j\in \mathcal{J}_n , j \geq 2^{n-1}\}$. In words, $\overline{\mathcal{J}_n}$ is the subset of $\mathcal{J}_n$ that is on the upper-half of the leftmost stage within the factor graph, and $\underline{\mathcal{J}_n}$ is on the lower-half. Similarly, for each $j \in \mathcal{J}_{n}$ such that $j \geq 2^{n-1}$, we define $\underline{j} = j - 2^{n-1}$. Clearly, we have $f(\underline j) = \frac{1}{2} f(j)$ (For detailed proof, see Fact 6 in \cite{Eslami2013}).

\begin{proposition}
    If $\mathcal{J}_{n+1}$ satisfies both conditions, then $\underline{\mathcal{J}_{n+1}}$ would also satisfy both conditions. 
\label{underset}
\end{proposition}

\begin{proof}
Suppose $j^* \in \underline{\mathcal{J}_{n+1}}$,
   $i^*$ covers $j^*$, and $j^* \nearrow i^*$. According to the definition of $\underline{\mathcal{J}_{n+1}}$, we know that $j^* + 2^n \in \mathcal{J}_{n+1}$. 

   Cover condition: Since $i^* + 2^n$ covers $j^* + 2^n$, and $j^* + 2^n \in \mathcal{J}_{n+1}$, we know that $i^* + 2^n \in \mathcal{J}_{n+1}$. Thus $i^* \in \underline{\mathcal{J}_{n+1}}$.

    Swap condition: Since $j^* + 2^n \nearrow i^* + 2^n$, and $j^* + 2^n \in \mathcal{J}_{n+1}$, we know that $i^* + 2^n \in \mathcal{J}_{n+1}$. Thus $i^* \in \underline{\mathcal{J}_{n+1}}$.


    
\end{proof}

\begin{thm}
$|MVSS(\mathcal{J})| = \min\limits_{i \in \mathcal{J}} f(i)$ if set $\mathcal{J}$ satisfies the cover condition and the swap condition.
\label{th2}
\end{thm}

\begin{proof}
    See Appendix.
\end{proof}

In fact, any decreasing monomial set (see Definition 4 in \cite{Bardet2016}) satisfies both the cover and swap conditions, as can be directly inferred from the definition and the relationship between the row indices and their corresponding monomials. The information set $\mathcal{A}$ of a polar code is known to be a decreasing monomial set \cite{Bardet2016}, and thus satisfies both conditions. This is also confirmed in \cite{PO2016}, where it is shown that the polar information set $\mathcal{A}$ meets these conditions, leading to the following corollary:

\begin{corollary}
    For a polar code with information set $\mathcal{A}$, $|MVSS(\mathcal{A})| = \min\limits_{i \in \mathcal{A}} f(i)$.
\label{c1}
\end{corollary}

Consider a binary erasure channel (BEC), where the values of the variable nodes can be 0, 1, or erasure, and a belief propagation (BP) decoder. We examine the erasures on the information nodes after sufficient iterations of BP decoding. Corollary~\ref{c1} highlights an interesting fact: the probability of erasure for the entire set $\mathcal{A}$ is greater than or equal to the probability of erasure for a single bit $MIB(\mathcal{A})$. The strict inequality may arise when there is more than one MVSS for the set $\mathcal{A}$.


\section{Outer polar stopping set construction}

The proposed bounds suggest a practical way to design outer polar codes based on the size of stopping sets. In this section, we focus on Deletion Bound I as an example, given its strong performance among all the proposed bounds and its deterministic nature (independent on random seeds). Similar constructions can be easily extended to other proposed bounds.

\subsection{Construction method}

Denote by $d(i)$ the upper bound of $|MVSS(\mathcal{H}_i)|$ obtained by Deletion Bound I. We first initialize an unfrozen set $\mathcal O$ for the outer code using the conventional DE, for example. Then we swap a specified number of unfrozen bits $i\in{\mathcal O}$ with the smallest ``stopping distance'' $d(i)$ with some positions $j \in \mathcal{O}^c$ such that $d(j) > d(i)$. 

Before presenting the design algorithm, we introduce some notations. Let $Q$ be a length $N_0$ vector that contains the indices of bit-channels ordered according to channel reliability calculated by DE. The indices are ordered by descending channel reliability, i.e., $Q(1)$ stores the index for the strongest bit channel, $Q(2)$ stores the index for the second strongest, and so on. Let $s$ denote the number of bits we are going to swap. Let $K_0$ denote the size of the desired unfrozen set. Let $min_s(\cdot)$ be the function that returns the $s$-th smallest value in a vector, while $min(\cdot)$ returns the smallest value along with its index. Note that $s$ should be chosen such that there are more than $s$ frozen bits that have $d(\cdot)$ value larger than $min_{s}(d(Q(1)),...,d(Q(K_0)))$. The detailed swapping algorithm is presented in Algorithm~\ref{alg1}. Note that Algorithm~\ref{alg1} can be easily extended to incorporate other bounds by setting $d(i) = |MVSS(\mathcal{H}_i)|$, where $|MVSS(\mathcal{H}_i)|$ represents the value determined by the respective bounds. 


\begin{algorithm}
\caption{Outer polar stopping set (OPSS) construction}\label{alg:cap}
\begin{algorithmic}[1]
\Statex \textbf{Input:} $Q$; $d(i)$ for each $i<{N_0}$; $s$
\Statex \textbf{Output:} designed unfrozen set $\mathcal O$

\State $threshold = min_{s}(d(Q(1)),...,d(Q(K_0)))$
\State $i \gets 1$
\While{$i \leq s$}
    \State $[value, index] = min(d(Q(1)),...,d(Q(K_0)))$
    \State $j \gets 1$
    \While{True}
        \If{$d(Q(K_0+j)) > threshold$}
            \State $Q(index) \gets Q(K_0+j)$
            \State delete $Q(K_0+j)$ from $Q$
            \State jump to line 14
        \EndIf
        \State $j \gets j+1$
    \EndWhile
    \State $i \gets i+1$
\EndWhile
\State Return $\mathcal{O} = Q(1:K_0)$
\end{algorithmic}
\label{alg1}
\end{algorithm}

We can easily extend Deletion Bound I to the case when $M$ inner codes are connected by a single outer code. For example, assume there are $M=2$ inner codes and $\mathcal{H}_i = \{ \mathcal{H}_i^1, \mathcal{H}_i^2 \}$, where $\mathcal{H}_i^1$ and $\mathcal{H}_i^2$ are connected nodes in the first and second inner codes, respectively. Then $d(i) = |MVSS(\mathcal H_i^1)| + |MVSS(\mathcal H_i^2)|$.

The design method of Algorithm~\ref{alg1} can be extended to the local-global polar code architecture, but some care is needed. The systematic outer polar code assigns $M$ information vectors $K_{a_i}$, $i{=}1, \ldots, M$ to the $M$ inner codes. Directly applying Algorithm~\ref{alg1} can potentially swap bits $K_{a_i}$ with $P_{a_j}$ ($i \neq j$), causing $[K_{a_i},K_{a_j}]$ to be assigned to the same inner code. For example, assume the unfrozen set $\mathcal{O} = \{ 2,3,6,7\}$ represents the most reliable positions according to DE, and $\mathcal{O}^c = \{ 0,1,4,5\}$. Then, if the partition of $\mathcal{O}$ is according to bit index, the first half of $\mathcal{O}$ will correspond to $K_{a_1}=\{ 2,3\}$ and the second half to $K_{a_2}=\{ 6,7\}$. If the parity bits are partitioned similarly, we have $P_{a_1}=\{ 0,1\}$ and $P_{a_2}=\{ 4,5\}$. If, after calculating the $d(i)$ value for each position, we swap positions 2 and 4, this would yield $K_{a_1}=\{ 3,4\}$. This assignment is now inconsistent with the local-global architecture because part of $K_{a_1}$ (position 4) is connected with the second inner code. To avoid this problem, one needs to carefully design the partition of the outer codeword in the local-global encoder to ensure that positions in $K_{a_i}$ are only swapped with positions in $P_{a_i}$. 

\begin{example}
\label{localglobal}
For the case $M=2$, there is a mapping and partition rule that works for any $s \leq 4$, $2^9 \leq N_1 = N_2 \leq 2^{11}$, $2^6 \leq N_{0} \leq 2^8$ with $R_{0} = R_{all} = \frac{1}{2}$. Let $\mathcal{O}$ be the information set of size $N_{0}/2$ according to DE. Assign to $K_{a_2}$ the first half according to natural index order, and to $K_{a_1}$ the second half. Assign to $P_{a_1}$ the first half of the parity positions in natural index order, and to $P_{a_2}$ the second half. The semipolarized bit-channels in the first inner code are connected in natural index order to $K_{a_1}, P_{a_1}$ and those in the second inner code connect to $K_{a_2}, P_{a_2}$. See Fig.~\ref{Modified_encoder}.
\end{example}

\begin{figure}[htbp]
\centerline{\includegraphics[width=8cm,height=4.2cm]{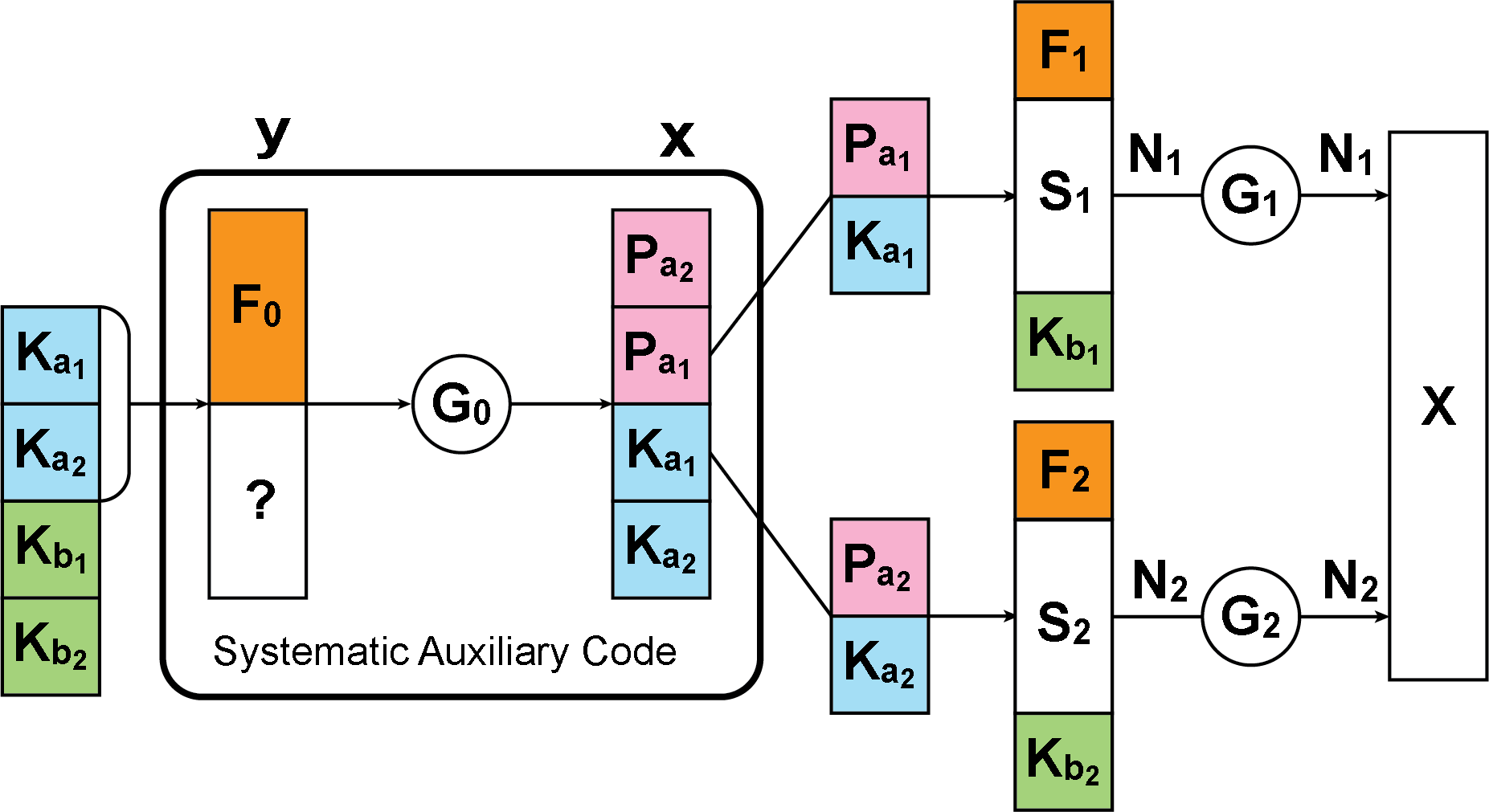}}
\caption{Local-global encoder for OPSS construction.}
\label{Modified_encoder}
\end{figure}

\subsection{Experimental results}

Now we give empirical results under BP decoding for augmented and local-global polar codes. The bit-channel ordering is based on DE on the AWGN channel at $E_b/N_0=3$ dB. DE is simplified by using Gaussian approximation (GA), using the 4-segment approximation function in~\cite{GA Polar}. The BP decoding schedules are the same as those in \cite{Elk2017} for augmented codes and in \cite{Zhu2022} for local-global codes. The maximum number of BP decoder iterations is set at 100. In the OPSS design (Algorithm 3), we set the number of bit-channel swaps to $s=4$. We also include the results for outer codes designed by non-stationary density evolution (NDE), which is another outer code construction method proposed by us in the shorter version of this paper \cite{Zhu2024}. The NDE algorithm does not assume that the inital LLR distributions to the outer code are identical, as assumed by conventional DE; rather, the LLR distributions coming from the inner code at the rightmost stage of the outer code factor graph correspond to $N_0$ separate binary symmetric memoryless channels $W_i$, $i{=}0,\ldots,N_0{-}1$. In practice, we replace each initial LLR density of the outer code with the empirical LLR density of the corresponding bit-channel of the inner code after $ite$ iterations of BP decoding under assumption of an all-zero codeword. In the NDE design for the augmented code, we use $ite=3$ BP decoder iterations on the inner code to generate the required empirical LLRs, and for the local-global code, we use $ite=4$ iterations.

Fig.~\ref{FER_augmented} shows frame error rate (FER) results for augmented code constructions. The outer code length is $N_0=64$ with code rate $R_0=\frac{1}{2}$. The inner code length is $N_1=1024$. The design rate of the augmented code is $R_{aug} = \frac{1}{2}$. The connections between the bit-channels of the inner code and the bits of the outer codeword are based on the natural index ordering within the set of semipolarized bit-channels. 

At FER = $10^{-3}$, the OPSS design and NDE design offer gains of 0.12 dB and 0.18 dB over the conventional DE design, respectively. At this FER they also perform comparably to SCL decoding, although SCL becomes superior at higher FERs. We remark that when the outer codes designed using the OPSS and NDE methods are disconnected from the concatenation architecture, their performance is inferior to that of a code designed using conventional DE. This confirms their inherent relationship with the concatenation structure.  

Figs.~\ref{local} and~\ref{global} present the results for local and global decoding, respectively, for a local-global code with component code lengths $N_0=256 ,  N_1 = N_2 = 1024$. The connections between the inner codes and the outer code are as  described in Example~\ref{localglobal}. Local decoding results for the different outer code design methods are similar, as expected, since local decoding does not rely on the outer code. Under global decoding, at FER = $10^{-4}$, the OPSS and NDE designs provide gains of 0.19 dB and 0.07 dB over conventional DE, respectively. The result for a length-2048 conventional polar code is also shown for reference. In summary, the improved global decoding performance provided by the new outer code constructions does not reduce the local decoding performance. 

\begin{figure}[htbp] 
\centerline{\includegraphics[width=7cm,height=5.5cm]{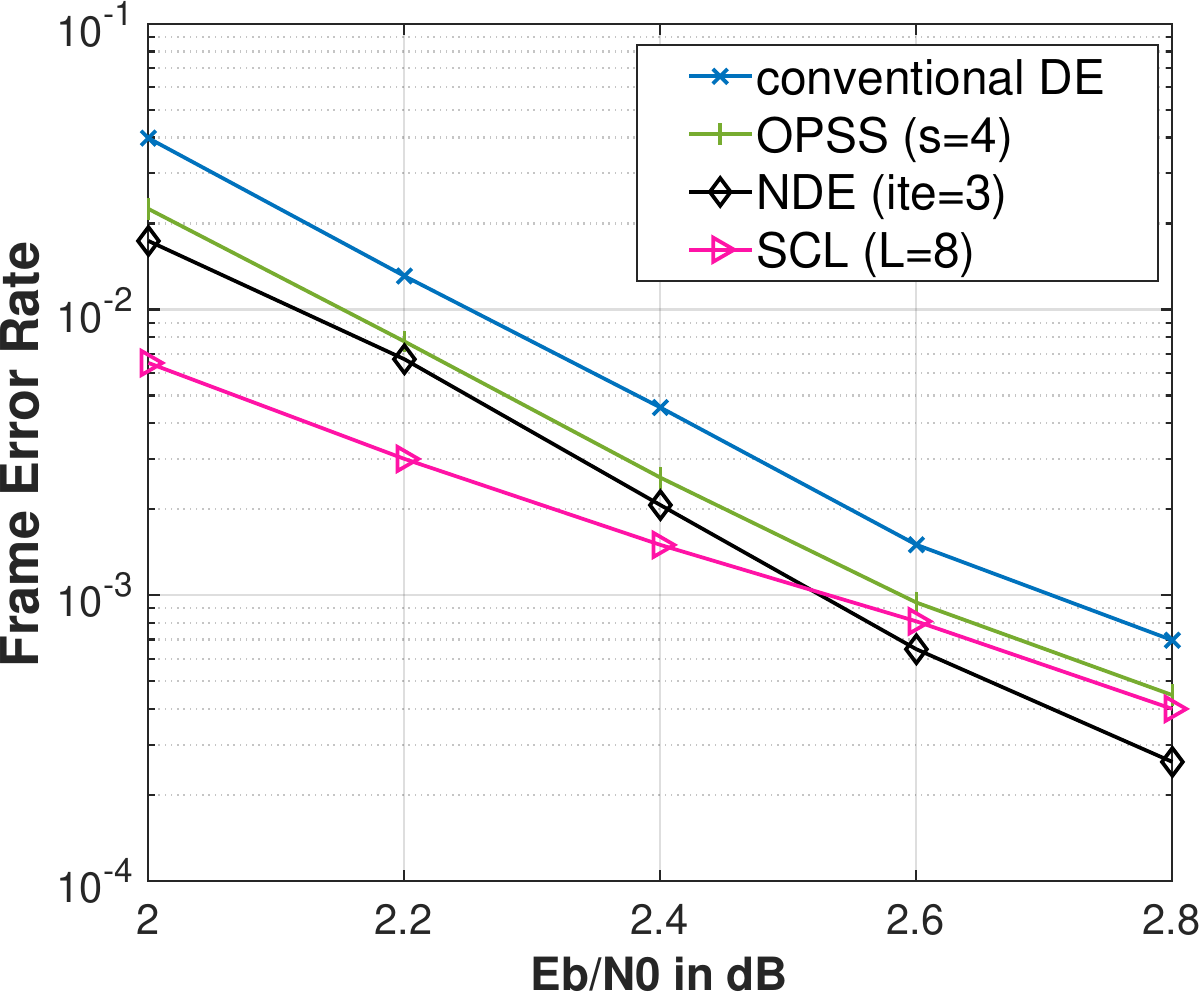}}
\caption{Augmented code $N_0{=}64$, $N_1{=}1024$ with natural interleaver.}
\label{FER_augmented}
\end{figure}

\begin{figure}[htbp] 
	\centering  
	\vspace{-0.35cm} 
	\subfigtopskip=4pt 
	\subfigbottomskip=4pt 
	\subfigcapskip=2pt 
	\subfigure[local decoding results]{
		\label{local}
		\includegraphics[width=7cm,height=5.65cm]{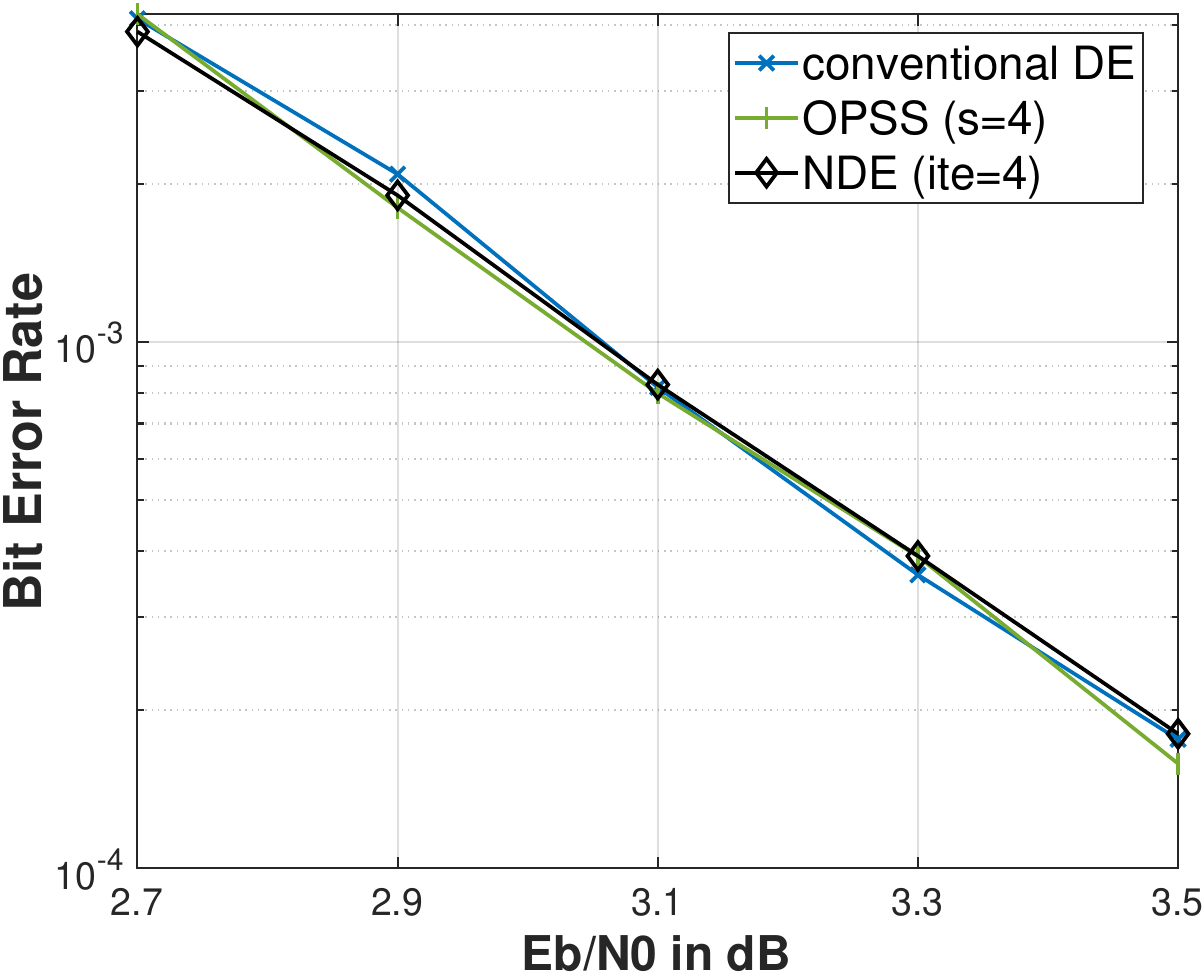}}
	\quad 
	\subfigure[global decoding results]{
		\label{global}
		\includegraphics[width=7cm,height=5.65cm]{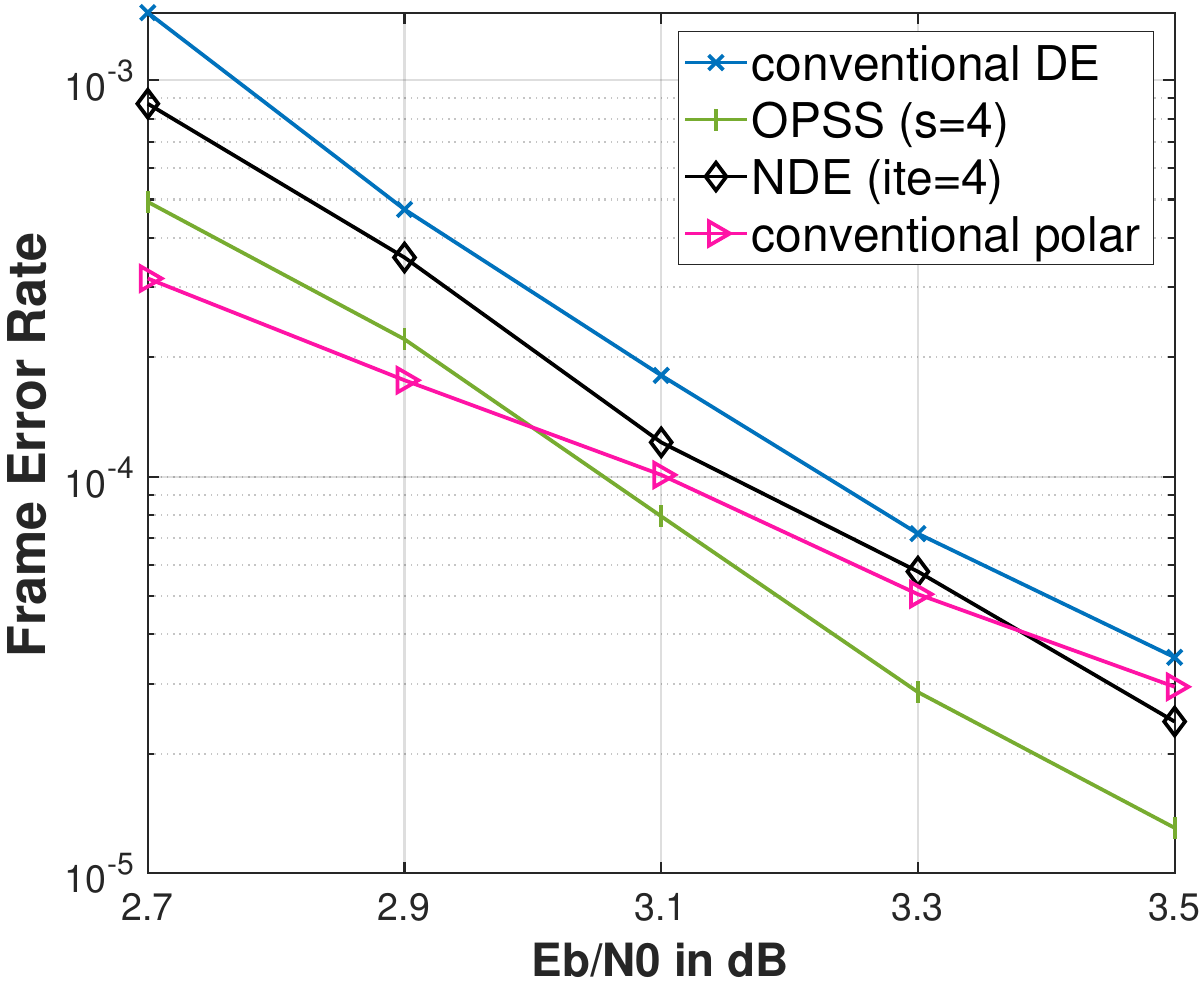}}
	  \\
        \caption{Local-global code with $N_0=256, N_1=N_2=1024$.}
\end{figure}

Although the natural ordering is attractive from an implementation standpoint, it may not provide the best starting point for these design methods, and experiments with other connection patterns, both structured and randomly generated, show that gains achieved with the proposed methods vary. For example in Fig.~\ref{difpattern}, we selected four different patterns for augmented codes whose performance are close to the average of all the interleaver patterns, within which pattern 1 and 3 benefit from OPSS. These results suggest that while some interleavers benefit from the proposed method, others might not experience the same level of enhancement. This highlights the importance of selecting an appropriate outer code design method that aligns well with the connection pattern. How to jointly optimize the interleaver pattern and the design methods remains a problem for further research. 

\begin{figure}[t] 
	\centering  
	\vspace{-0.35cm} 
	\subfigtopskip=2pt 
	\subfigbottomskip=2pt 
	\subfigcapskip=-5pt 
	\subfigure[pattern 1]{
		\includegraphics[width=4.05cm,height=3cm]{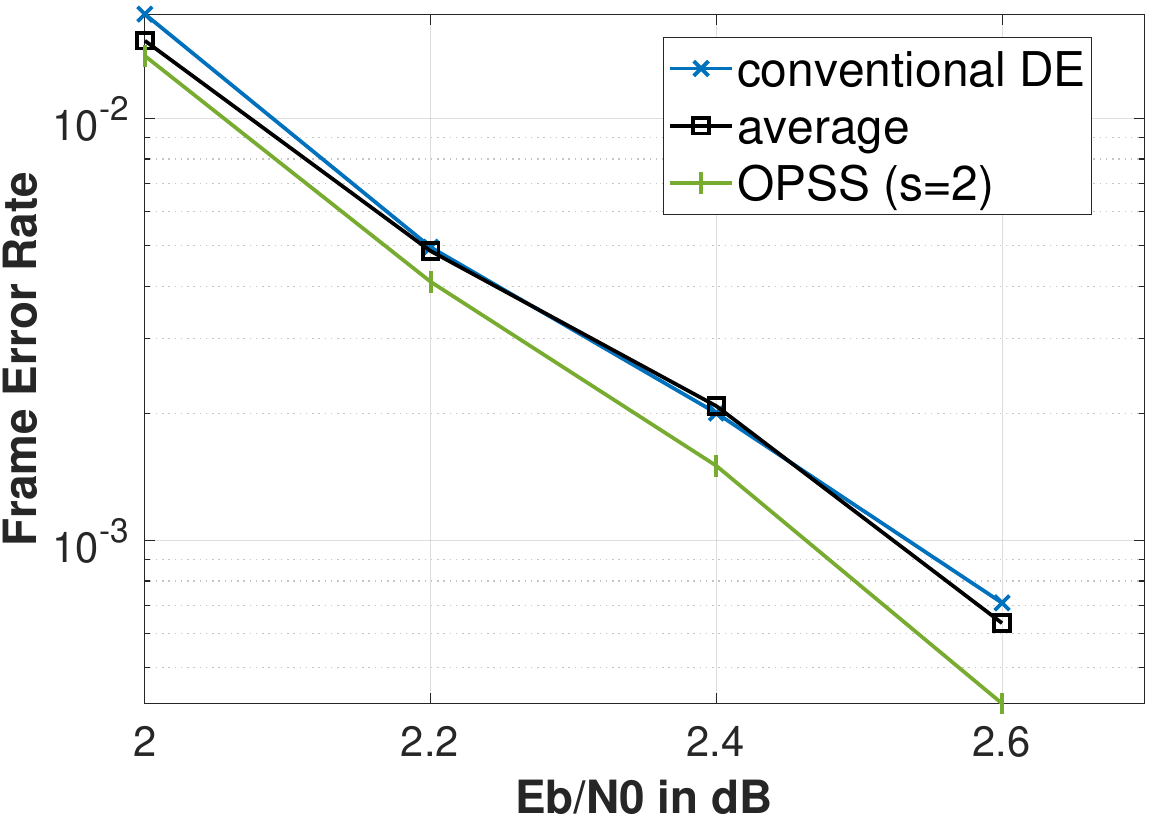}}
	\quad 
	\subfigure[pattern 2]{
		\includegraphics[width=4.05cm,height=3cm]{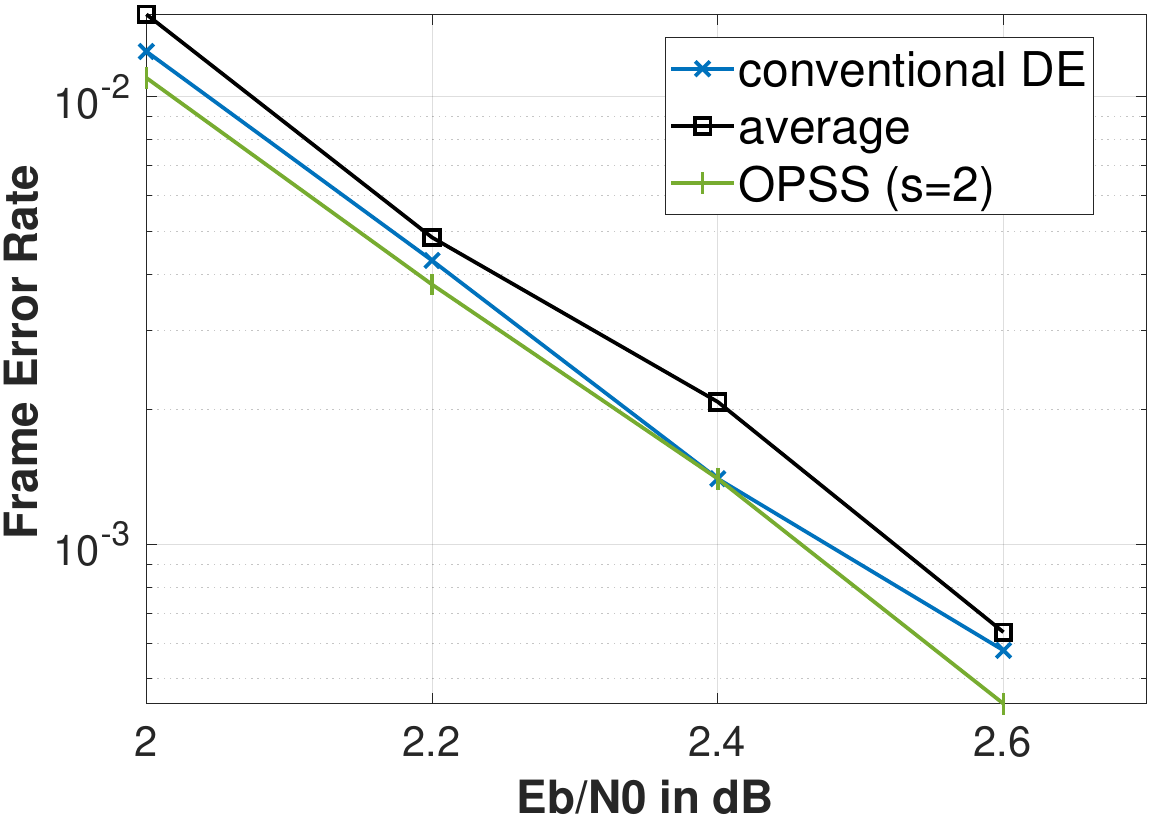}}
	  \\
	\subfigure[pattern 3]{
		\includegraphics[width=4.05cm,height=3cm]{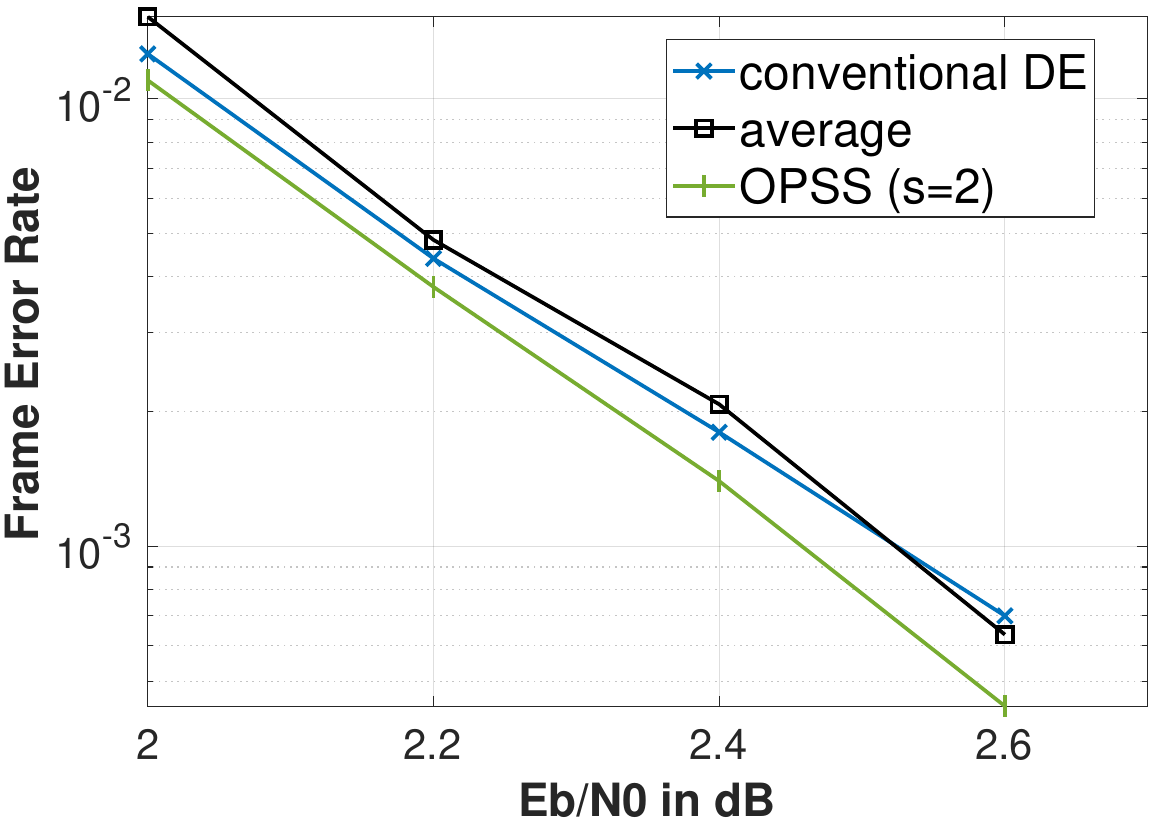}}
	\quad
	\subfigure[pattern 4]{
		\includegraphics[width=4.05cm,height=3cm]{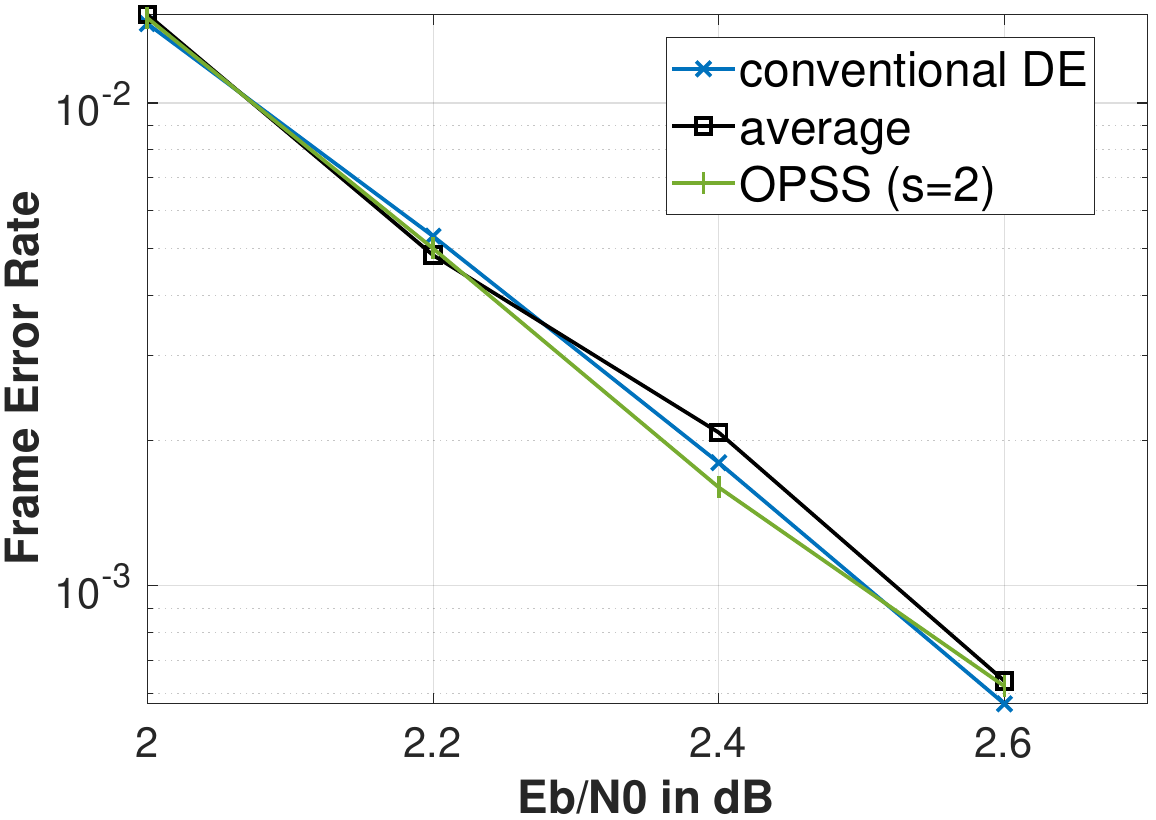}}
	\caption{Augmented code $N_0=64, N_1=1024$ with different interleaver patterns.}
	\label{difpattern}
 \vspace{-0.6cm} 
\end{figure}

\section{Conclusion}

In this paper, we proposed four bounds on the value of $|MVSS(\mathcal{J})|$, which can be used to bound the stopping distance of concatenated polar codes. We proposed a design method for the outer code in augmented and local-global polar code architectures based on the “stopping distance” associated with each bit position.
However, several questions still remain to be answered:

\begin{itemize}
    \item In Section III, a practical method for computing of $|MVSS'(\mathcal{J}_{out})|$ remains to be found. 
    \item In Section IV, a practical method for computing $|MVSS(\mathcal{J})|$ is still remains to be found. A potential direction for future research could involve proving the NP-hardness of this problem.
    \item In Fig.~\ref{Polar1024}, we only proved the tightness for Lower Bound~I. However, it appears that Deletion Bound~I is also tight under certain conditions, as indicated by the following conjecture.
\end{itemize}

\begin{conj}
    If set $\mathcal{J}$ satisfies both cover condition and swap condition, then Algorithm~\ref{algo_version1} (Deletion Bound I) can find the exact value of $|MVSS(\mathcal{J})|$.
\label{conjecture}

\end{conj}

\appendix

\textbf{Proof of Theorem~\ref{th1}}: The definition of stopping set implies that the union of stopping sets is a stopping set. For a given set of information bits $\mathcal{J}$, consider set $UT(\mathcal{J})$, the union of all the stopping trees defined by the elements in $\mathcal{J}$. To find a minimum size VSS for $\mathcal{J}$, we can try to find a stopping set that is properly contained in $UT(\mathcal{J})$ that preserves $\mathcal{J}$ but has fewer leaf nodes (i.e., observed variable nodes). According to Proposition~\ref{nOLL}, the only leaf nodes that can possibly be deleted from $UT(\mathcal{J})$ are overlapped leaf nodes (OLLs), i.e., leaf nodes that belong to at least two stopping trees in $UT(\mathcal{J})$. Therefore, deleting the entire set of $OLL(\mathcal{J})$ would establish a lower bound.

To complete the proof, we need to characterize the indices of the leaf nodes in $UT(\mathcal{J})$. This can be done by noting that the
indices for the leaf nodes in the stopping tree $ST(i)$ are given by the indices for the ones in $r_i^n$, where $r_i^n$ is the ($i+1$)-th row of $G^n =F^{\bigotimes n}$, the encoding matrix for length $2^n$ polar codes. This fact was stated in~\cite{Eslami2013} without proof. For completeness, we provide a detailed proof here. The proof proceeds by induction. For the case $n=1$, the statement follows immediately from inspection of the matrix $G^1=F$ and inspection of the corresponding factor graph for $n=1$.  

Now, suppose the result is true for a given $n$. For a length $2^{n+1}$ polar code, the recursive construction of the factor graph implies that $T_{n+1}^{U}$ and $T_{n+1}^{L}$ are isomorphic, and $G^{n}$ is the encoding matrix for each of the subgraphs $T_{n+1}^{U}$ and $T_{n+1}^{L}$. Let $I_i^{n}$ be a length $2^n$ binary vector in which the  indices of ones are the positions of the leaf nodes contained in $ST(i)$ in $T_{n+1}^{U}$. By the induction hypothesis, $I_i^{n} = r_i^n$.  

Referring to Fig.~\ref{GSS}, we see that, for $0\leq i\leq 2^n-1$,  $I_i^{n+1} = [I_i^n,0,...,0] = [r_i^n,0,...,0] =  r_i^{n+1}$,
with the last equality following from~(\ref{recursive_G}). Similarly, for $2^n\leq i\leq 2^{n+1}-1$, since the subgraphs $T_{n+1}^{U}$ and $T_{n+1}^{L}$ are isomorphic, we have $I_i^{n+1} = [I_{i-2^n}^n,I_{i-2^n}^n] = [r_{i-2^n}^n,r_{i-2^n}^n]= r_i^{n+1}$, where again the last equality follows from~(\ref{recursive_G}). 
This completes the induction. 

In summary, the number of leaf nodes in $UT(\mathcal{J})$ is given by the number of columns in $G_{\mathcal{J}}$ that have non-zero weight. The number of nOLLs in $UT(\mathcal{J})$ is given by the number of columns in $G_{\mathcal{J}}$ that have weight exactly one. Thus $g(G_{\mathcal{J}})$ is precisely the number of unshared leaf nodes, which by the previous discussion must belong to a minimum size VSS for $\mathcal{J}$. This implies $|MVSS(\mathcal{J})|\geq g(G_{\mathcal{J}})$, as desired. 
\qed

\begin{figure}[htbp] 
	\centering  
	\vspace{-0.35cm} 
	\subfigtopskip=2pt 
	\subfigbottomskip=2pt 
	\subfigcapskip=-5pt 
	\subfigure[Example of case 1]{
		\label{case1}
		\includegraphics[width=8cm,height=6cm]{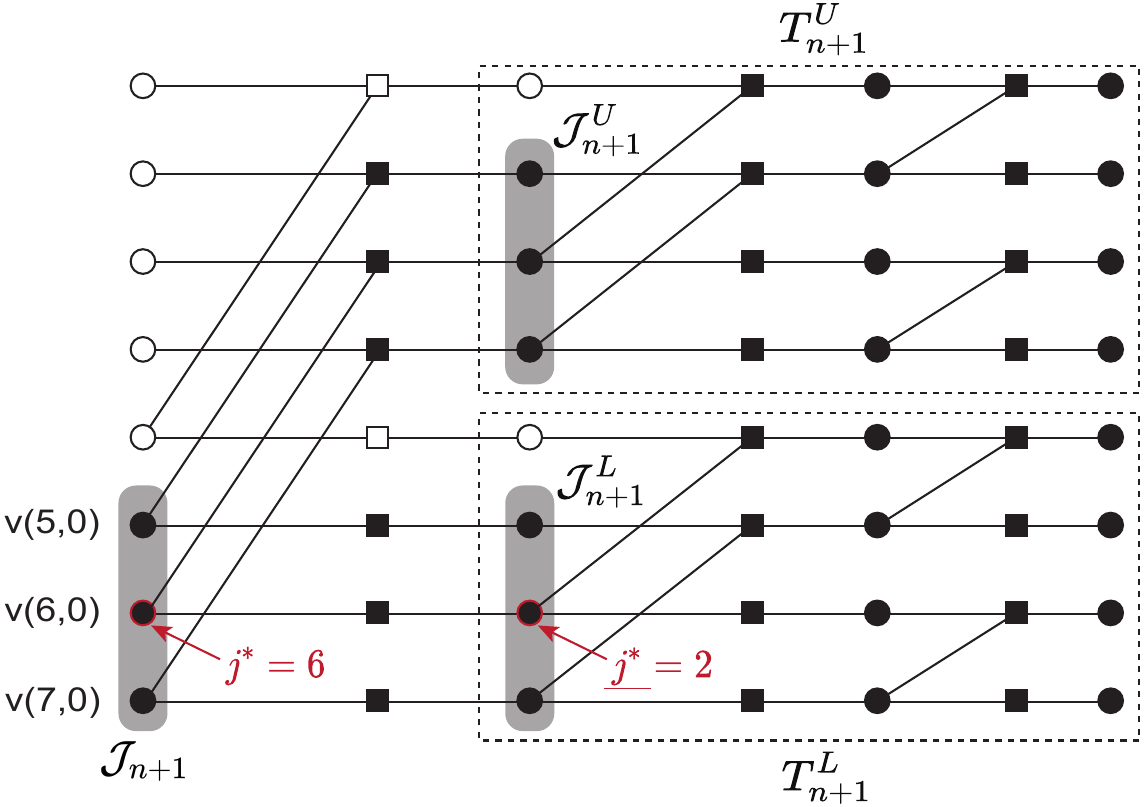}}
	\quad 
	\subfigure[Example of case 3]{
		\label{case3}
		\includegraphics[width=8cm,height=6cm]{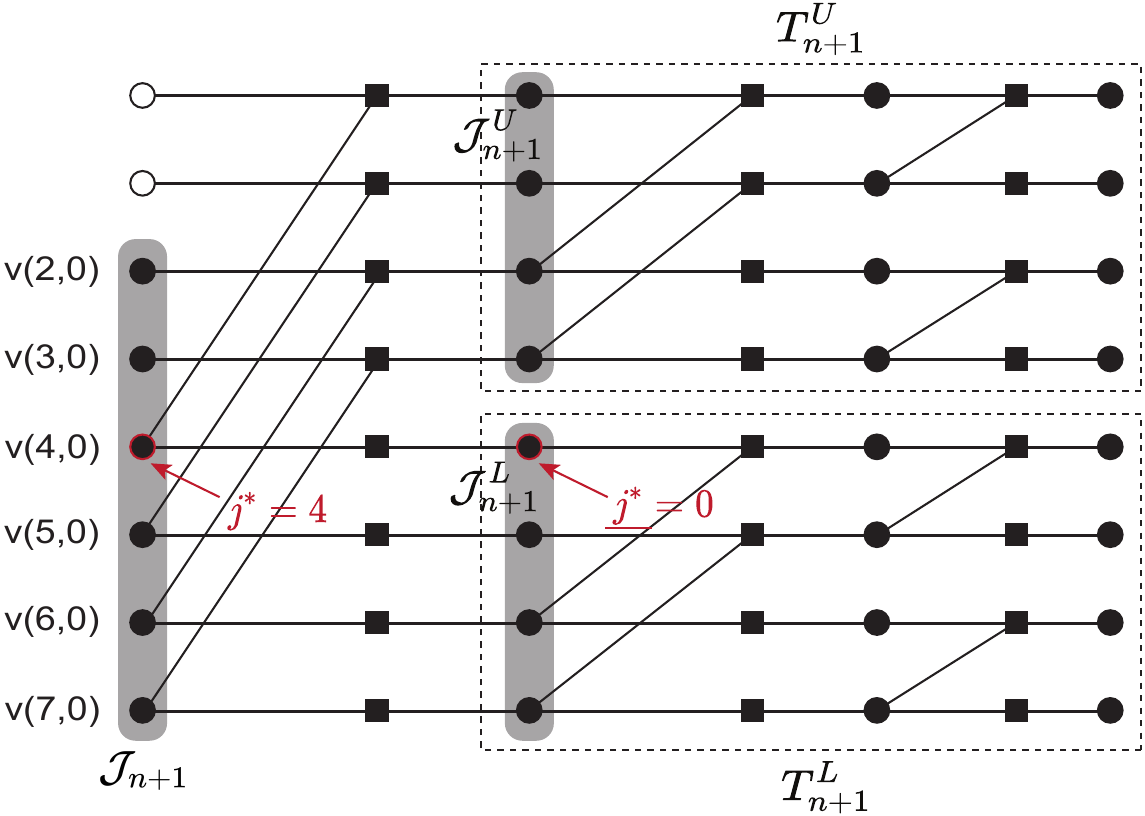}}
	  \\
        \caption{Illustrative examples for the proof of Theorem~\ref{th2}.}
\end{figure}
    
\textbf{Proof of Theorem~\ref{th2}}: The proof proceeds by induction. For the case $n=1$, the statement follows immediately from inspection of the corresponding factor graph. Now, suppose the result is true for a given $n$. By the induction hypothesis, $|MVSS(\mathcal{J}_n)| = \min\limits_{i \in \mathcal{J}_n} f(i)$ holds if $\mathcal{J}_n$ satisfies both conditions. Now, for a set $\mathcal{J}_{n+1}$ that also satisfies both conditions, consider the following cases:

1. $|\overline{\mathcal{J}_{n+1}}| = 0$. There is no information bit on the upper-half of the leftmost stage within the factor graph. Clearly $\mathcal{J}_{n+1}^U$ and $\mathcal{J}_{n+1}^L$ are identical, and $\mathcal{J}_{n+1}^U = \mathcal{J}_{n+1}^L = \underline{\mathcal{J}_{n+1}}$. We state that both of them would be in the stopping set corresponding to $MVSS(\mathcal{J}_{n+1})$, which is because all the neighbor check nodes in $UT(\mathcal{J}_{n+1})$ that are on the left of $\mathcal{J}_{n+1}^U$ and $\mathcal{J}_{n+1}^L$ are degree two. An example can be seen in Fig.~\ref{case1}, where $UT(\mathcal{J}_{n+1})$ is labeled by black nodes.

From Proposition~\ref{underset}, we know that $\mathcal{J}_{n+1}^U$ and $\mathcal{J}_{n+1}^L$ satisfy both conditions, thus 
\begin{equation*}
    |MVSS(\mathcal{J}_{n+1}^L)| = |MVSS(\mathcal{J}_{n+1}^U)| = \min\limits_{i \in \underline{\mathcal{J}_{n+1}}} f(i).
\end{equation*}

Now we prove that $\min\limits_{i \in \underline{\mathcal{J}_{n+1}}} f(i) = \frac{1}{2} \min\limits_{i \in \mathcal{J}_{n+1}} f(i)$. Assume $MIB^*(\mathcal{J}_{n+1}) = j^*$, then we state that $MIB^*(\underline{\mathcal{J}_{n+1}}) = \underline{j^*}$. Because if there exists $k \in \mathcal{J}_{n+1}$ such that $f(\underline{k}) < f(\underline{j^*})$, then $f(k) = 2 \times f(\underline{k}) < 2 \times f(\underline{j^*}) = f(j^*)$, which contradicts the assumption that $j^*$ is the MIB. Now, since $f(j^*) = 2\times f(\underline{j^*})$, we know that $|MVSS(\mathcal{J}_{n+1})| = 2\times |MVSS(\underline{\mathcal{J}_{n+1}})| = 2\times f(\underline{j^*}) = f(j^*)$.

2. $|\overline{\mathcal{J}_{n+1}}| \geq |\underline{\mathcal{J}_{n+1}}|$. There are more information bits on the upper-half than on the lower-half. We prove that this case is impossible unless $\mathcal{J}_{n+1}$ contains all the positions on the left (code rate is 1). According to the cover condition, any $i \in \overline{\mathcal{J}_{n+1}}$ would imply that $i + 2^n \in \underline{\mathcal{J}_{n+1}}$, thus $|\overline{\mathcal{J}_{n+1}}| \leq |\underline{\mathcal{J}_{n+1}}|$.

Now we assume that $|\overline{\mathcal{J}_{n+1}}| = |\underline{\mathcal{J}_{n+1}}|$, which implies that $j \in \mathcal{J}_{n+1} \Rightarrow \underline{j} \in \mathcal{J}_{n+1}$. We also assume that $j^*$ has the minimum weight (fewest number of 1's in the binary expression $(j^*)_b$) among all $j \in \underline{\mathcal{J}_{n+1}}$. Clearly $wt((j^*)_b) \geq 1$ since $j^* \geq 2^n$. We further assume that $wt((j^*)_b) > 1$. By the definition of $\underline{j^*}$, it directly follows that $wt((\underline{j^*})_b) = wt((j^*)_b) - 1 \geq 1$. We also know that the last bit of the binary representation of $\underline{j^*}$ is 0, i.e., $(\underline{j^*})_{n} = 0$. Without loss of generality, assume $(\underline{j^*})_m = 1, m < n$. Then, we can swap $(\underline{j^*})_m$ with $(\underline{j^*})_n$ to get a new index $l$. According to the swap condition, $l \in \mathcal{J}_{n+1}$ and since $l_n = 1$, we have $l \in \underline{\mathcal{J}_{n+1}}$. Now $wt(l_b)$ is smaller than $wt((j^*)_b)$, which contradicts the assumption. 

If $wt((j^*)_b) = 1$, then $j^* = 2^n$ and by the cover condition we know that $0 \in \mathcal{J}_{n+1}$, which implies that $\mathcal{J}_{n+1}$ contains all bits on the left. It is trivial to see that $|MVSS(\mathcal{J}_{n+1})| = \min\limits_{i \in \mathcal{J}_{n+1}} f(i) = f(0) = 1$.

3. $0 < |\overline{\mathcal{J}_{n+1}}| < |\underline{\mathcal{J}_{n+1}}|$. Let $MIB^*(\mathcal{J}_{n+1}) = j^*$. An example for this case is shown in Fig.~\ref{case3}. We first state that $j^* \geq 2^n$. Otherwise if $j^* < 2^n$, we can swap a 1 to the last bit to get another index $l$, such that $l > j^*$. From Proposition~\ref{weight_leaf}, we know that $f(l) = f(j^*)$ since $wt(l_b) = wt((j^*)_b)$. This contradicts the definition of $MIB^*$ as the MIB with the largest index. 

Again we have $\mathcal{J}_{n+1}^U = \mathcal{J}_{n+1}^L = \underline{\mathcal{J}_{n+1}}$, since any $j\in \overline{\mathcal{J}_{n+1}}$ would imply that $j + 2^n \in \mathcal{J}_{n+1}$, according to the cover condition. Similar to case 1, we can conclude that $MIB^*(\underline{\mathcal{J}_{n+1}}) = \underline{j^*}$, using the same proof. 

Now the difference with case 1 is: $\mathcal{J}_{n+1}^U$ might not be in the stopping set corresponding to $MVSS(\mathcal{J}_{n+1})$, since there are some degree-3 neighbor check nodes on the left of $\mathcal{J}_{n+1}^U$ so that some nodes in $\mathcal{J}_{n+1}^U$ may be excluded from the stopping set corresponding to $MVSS(\mathcal{J}_{n+1})$. Denote by $\mathcal{S}$ an arbitrary subset of $\mathcal{J}_{n+1}^U$. From Theorem~\ref{Bound Eslami}, we know that $|MVSS(\mathcal{S})| \geq \min\limits_{s \in \mathcal{S}} f(s)$. Since $\mathcal{S} \subseteq \mathcal{J}_{n+1}^U$, we have $|MVSS(\mathcal{S})| \geq \min\limits_{s \in \mathcal{S}} f(s) \geq \min\limits_{i \in \mathcal{J}_{n+1}^U} f(i) = f(\underline{j^*})$, which suggests that the minimum size we can expect in $T_{n+1}^U$ is actually $f(\underline{j^*})$. Then, the lower-half factor graph $T_{n+1}^L$ can be treated as incase 1, and we have $|MVSS(\mathcal{J}_{n+1})| = 2\times |MVSS(\underline{\mathcal{J}_{n+1}})| = 2\times f(\underline{j^*}) = f(j^*)$.

\qed

\ifCLASSOPTIONcaptionsoff
  \newpage
\fi


\section*{Acknowledgment}
This research was supported in part by NSF Grants CCF-1764104 and CCF-2212437. Furthermore, we would like to thank the anonymous reviewers for their valuable comments that helped to improve this paper.

\begin{IEEEbiography}[{\includegraphics[width=1in,height=1.25in,clip,keepaspectratio]{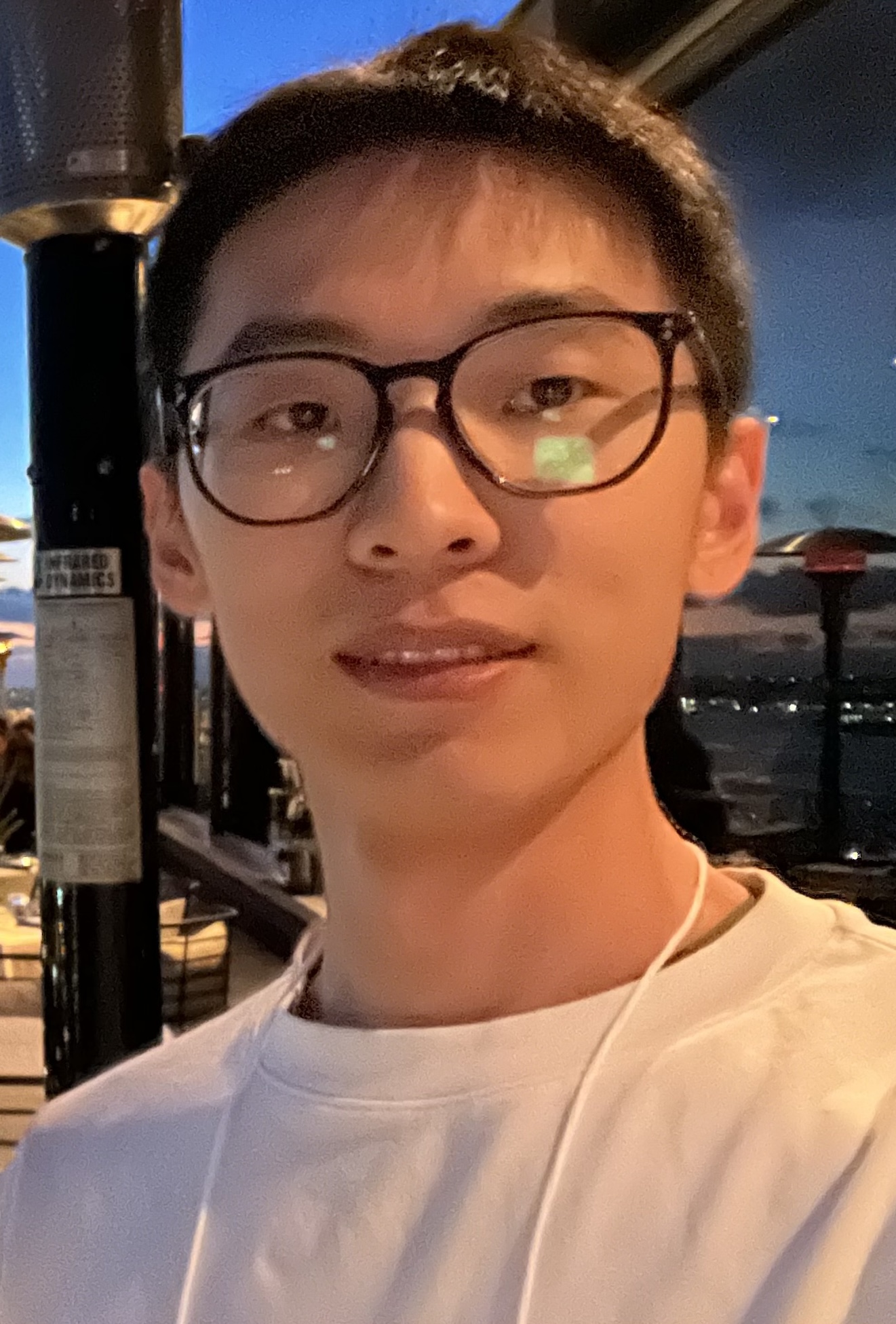}}]{Ziyuan Zhu} (Student Fellow, IEEE) received his B.E. in electrical communication engineering from Xi'an Jiaotong University, Xi'an, China, in 2021, and the M.S. in electrical and computer engineering from the University of California San Diego, San Diego, CA, USA, in 2023, where he is currently pursuing the Ph.D. degree with the department of electrical and computer engineering. He is also with the center for memory and recording research (CMRR), University of California San Diego. His research focuses on the application of error correction codes in storage systems, particularly polar and polar-like codes, with an emphasis on developing low-complexity decoding algorithms and efficient code constructions.
\end{IEEEbiography}

\begin{IEEEbiography}[{\includegraphics[width=1in,height=1.25in,clip,keepaspectratio]{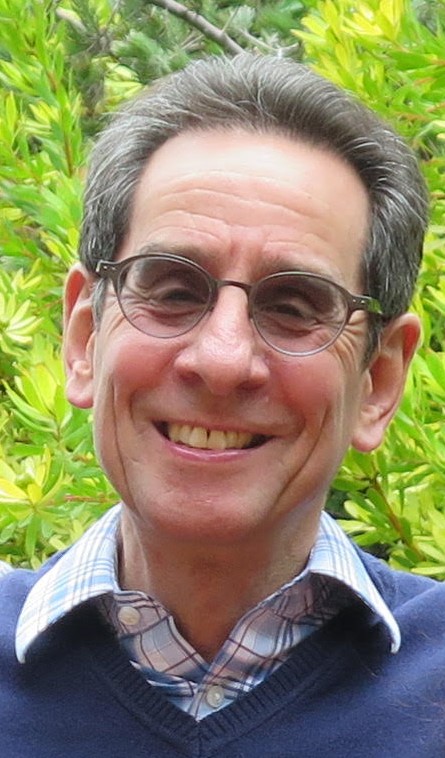}}]{Paul H. Siegel} (Life Fellow, IEEE) received the S.B. and Ph.D. degrees in Mathematics from the Massachusetts Institute of Technology, Cambridge, MA, USA, in 1975 and 1979, respectively. He held a Chaim Weizmann Postdoctoral Fellowship with the Courant Institute, New York University, New York, NY, USA. He was with the IBM Research Division, San Jose, CA, USA, from 1980 to 1995. He joined the faculty at the University of California San Diego (UCSD), La Jolla, CA, USA, in 1995, where he is currently a Distinguished Professor of Electrical and Computer Engineering with the Jacobs School of Engineering. He is affiliated with the Center for Memory and Recording Research where he holds an Endowed Chair and served as Director from 2000 to 2011. His research interests include information theory, coding techniques, and machine learning, with applications to digital data storage and transmission. He is a Member of the National Academy of Engineering. He was a Member of the Board of Governors of the IEEE Information Theory Society from 1991 to 1996 and from 2009 to 2014. He was the 2015 Padovani Lecturer of the IEEE Information Theory Society. He was a co-recipient of the 1992 IEEE Information Theory Society Paper Award, the 1993 IEEE Communications Society Leonard G. Abraham Prize Paper Award, and the 2007 Best Paper Award in Signal Processing and Coding for Data Storage from the Data Storage Technical Committee of the IEEE Communications Society. He served as an Associate Editor of Coding Techniques of the {\sc IEEE Transactions on Information Theory} from 1992 to 1995, and as the Editor-in-Chief from 2001 to 2004. He has served as Co-Guest Editor of special issues for the {\sc IEEE Transactions on Information Theory}, the {\sc IEEE Journal on Selected Areas in Communications} (twice), the {\sc IEEE Journal on Selected Areas in Information Theory}, the {\sc IEEE Transactions on Molecular, Biological, and Multi-Scale Communications},  and {\sc IEEE BITS the Information Theory Magazine}.
\end{IEEEbiography}

\end{document}